\documentclass[11pt]{article}
\usepackage{authblk} 
\usepackage[margin=1in]{geometry}  

\usepackage{amsmath,amsthm,amssymb}
\usepackage{booktabs}
\usepackage{multirow}
\usepackage{threeparttable}
\newtheorem{theorem}{Theorem}[section]
\newtheorem{lemma}{Lemma}[section]
\newtheorem{corollary}{Corollary}[section]
\newtheorem{remark}{Remark}[section]
\newtheorem{definition}{Definition}[section]

\newtheorem{example}{Example}[section]

\usepackage{setspace}              

\usepackage{rotating}
\usepackage{longtable}

\usepackage{natbib}
\bibpunct[, ]{(}{)}{,}{a}{}{,}%
\usepackage{pgfplots}

\usepackage{color,hyperref}
\definecolor{darkblue}{rgb}{0.0,0.0,0.6}
\hypersetup{colorlinks,breaklinks,linkcolor=darkblue,urlcolor=darkblue,anchorcolor=darkblue,citecolor=darkblue}
\usepackage{algpseudocode}
\usepackage{algorithm}

\usepackage{algpseudocode}
\usepackage[bb=dsserif]{mathalpha}
\usepackage{bm}

\title{\bf Compounding Effects in Leveraged ETFs: \\ Beyond the Volatility Drag Paradigm\thanks{This manuscript is a preprint and is currently under review for publication.}}

\author[1]{Chung-Han Hsieh}
\author[1]{Jow-Ran Chang}
\author[1]{Hui Hsiang Chen}

\affil[1]{\small Department of Quantitative Finance, National Tsing Hua University, Hsinchu, Taiwan, 30004}

\date{}
 
\begin{document}

\maketitle

\begin{abstract}
A common belief is that leveraged ETFs (LETFs) suffer long-term performance decay due to \emph{volatility drag}. We show that this view is incomplete: LETF performance depends fundamentally on return autocorrelation and return dynamics. In markets with independent returns, LETFs exhibit positive expected compounding effects on their target multiples. In serially correlated markets, trends enhance returns, while mean reversion induces underperformance. With a unified framework incorporating AR(1) and AR-GARCH models, continuous-time regime switching, and flexible rebalancing frequencies, we demonstrate that return dynamics---including return autocorrelation, volatility clustering, and regime persistence---determine whether LETFs outperform or underperform their targets. Empirically, using about 20 years of SPDR S\&P~500 ETF and Nasdaq-100 ETF data, we confirm these theoretical predictions. Daily-rebalanced LETFs enhance returns in momentum-driven markets, whereas infrequent rebalancing mitigates losses in mean-reverting~regimes.
\end{abstract}

\section{Introduction}\label{section: Introduction}
Exchange-traded funds (ETFs) have transformed the investment landscape, offering investors cost-effective, liquid, and diversified exposure to various asset classes, e.g., see \cite{malkiel2001growth, ben2017exchange}.  As of December 2023, global ETF assets under management surpassed \$11.5 trillion, driven by an annual compound growth rate of 18.9\% over the preceding five years.\footnote{According to the survey published by PwC (ETFs 2028: Shaping the future), assets under management of global ETFs experienced a compound annual growth rate (CAGR) of 18.9\% over the last five years. Moreover, they increased by over 25\% since December~2022, reaching a new milestone of nearly \$11.5 trillion by the end of 2023.} See also

Leveraged ETFs (LETFs), specialized products within this market,  typically implement a \emph{constant daily leverage} strategy, e.g., see \cite{trainor2008leveraged, leung2016leveraged}, ensuring that leverage exposure is reset at the end of each trading session. For example, a 2x LETF seeks to deliver twice the daily return of its benchmark index. In practice, this leverage is typically achieved through derivatives such as total return swaps (TRS) and futures contracts; \cite{leung2016leveraged}.

Despite their rapid growth, LETFs remain controversial due to concerns over long-term performance decay from \emph{volatility drag}, which refers to the phenomenon where the compounded (geometric) return of an LETF is less than its arithmetic average return due to the effects of volatility; e.g., see \cite{avellaneda2010path, trainor2012volatility, Hongtang2013solving}. See also \cite{SEC2009} for an alert announcement from SEC regarding the riskiness of LETFs.
Prior studies have primarily focused on volatility drag as the dominant driver of LETF underperformance, e.g., see \cite{jarrow2010understanding, charupat2011pricing, trainor2017leveraged}, often emphasizing the negative effects of rebalancing on long-term~returns.

However, we argue that this conventional view is incomplete.
The literature has placed less emphasis on aspects of return behavior that fundamentally alter LETF performance dynamics. For example,
\cite{lu2009long} showed empirically that LETFs can meet their leverage targets over horizons shorter than a month but may deviate when held beyond a quarter. \cite{andrewcarver2009leveraged} pointed out that the long-term returns of high-leveraged ETFs may decay toward zero, as indicated by the concept of growth-optimal portfolios. Further studies, such as \cite{guedj2010leveraged}, demonstrated that LETF investors could incur losses even when the underlying index exhibits positive returns, emphasizing the risks of long-term LETF holdings.
Recently, \cite{dai2023leveraged} studied an optimal rebalancing model for LETFs under market closure and frictions.
Yet these analyses fail to fully characterize the mechanism driving LETF performance.

Throughout this paper, we use the term \emph{return dynamics} to refer broadly to the time-series structure of returns---including independent returns, serial dependence (e.g., autocorrelation), volatility clustering (e.g., GARCH effects), and regime-driven variation in drift and volatility. This encompasses both stationary i.i.d. settings and more complex dynamic models such as AR-GARCH and regime-switching processes

We demonstrate that return autocorrelation and return dynamics play a primary role in driving compounding deviations under realistic return dynamics.  
Specifically, in markets with independent returns, we show that LETFs exhibit positive expected compounding effects to their target multiples, contrary to conventional wisdom.  In serially correlated markets, trends enhance returns while mean reversion induces underperformance. These findings challenge the established paradigm that volatility drag alone dictates LETF performance.

Our theoretical framework and empirical validation over approximately 20 years, spanning multiple LETFs and market crises, confirm that LETF performance cannot be fully explained by volatility drag alone. Rather, it is fundamentally influenced by return autocorrelation, volatility clustering, and regime persistence---factors less emphasized in prior research that relies solely on realized historical volatility. This comprehensive analysis resolves the apparent contradictions in existing literature and provides a more nuanced understanding of when LETFs outperform or underperform their targets under various market conditions.

\subsection{Key Contributions and Paper Organization}
The key contributions of this paper are summarized as follows:

\begin{itemize}
	\item Return Dynamics as Primary Driver of LETF Performance: We demonstrate that return autocorrelation and return dynamics, not simply volatility drag, play a primary role in driving compounding deviations under realistic return dynamics. Our analysis provides sufficient conditions for LETF outperformance in trending markets and underperformance in oscillating regimes, resolving contradictions in previous empirical findings.
	
	\item Unified Theoretical Framework: We develop a unified framework incorporating discrete-time and continuous-time regime switching dynamics, allowing precise characterization of LETF behavior across diverse market conditions. 
	
	\item  Empirical Validation: Using approximately 20 years of historical data from S\&P 500 ETF (Ticker: \texttt{SPY}) and Nasdaq-100 ETF (Ticker: \texttt{QQQ}), we demonstrate that the associated LETF performance systematically outperforms in momentum regimes and underperforms during mean-reverting periods.
	
	\item Practical Guidance for Investors: Our findings reveal that daily-rebalanced LETFs are optimal in momentum markets, while weekly/monthly rebalancing reduces losses in mean-reverting conditions. This insight is critical for hedge funds, portfolio managers, and individual traders optimizing leveraged positions.
	
\end{itemize}

The remainder of this paper is organized as follows.  Section~\ref{section: Preliminaries} develops the theoretical framework for analyzing the compounding effect under different return structures. Section~\ref{section: Compounding Effect Analysis} presents both theoretical analysis and simulations to evaluate LETF performance. Subsequently,  in Section~\ref{section: Empirical Studies},  we validate our theoretical findings against historical data from \texttt{SPY} and \texttt{QQQ},  covering multiple market regimes such as the 2008 financial crisis and the COVID-19 pandemic. Each section integrates theoretical insights with practical implications,  and concludes in Section~\ref{section: conclusion}.

\section{Framework for Analyzing LETF Behavior}\label{section: Preliminaries}
This section introduces the concept of \emph{compounding effect}, which plays a central role in understanding LETF performance deviations.

\subsection{Compounding Effect on LETF}\label{variable definition and methodlogy}
Since LETFs rely on daily rebalancing to maintain a \emph{fixed} leverage ratio, they reflect the target leverage only on a daily basis. Over longer holding periods, compounding effects cause the cumulative returns of LETFs to deviate from their expected multiple, a phenomenon extensively documented in prior studies, e.g., see \cite{avellaneda2010path}. This deviation reveals the importance of understanding LETF performance beyond short-term horizons.

\medskip
\begin{example}[Daily Compounding and Effective Leverage] \rm
	Consider an index that rises by 6\% on the first day and falls by 4\% the next day; the two-period cumulative return is 1.76\%. Now, if an investor holds a 2x leveraged ETF based on this index, they will initially gain 12\% in returns on the first day but then experience an 8\% loss on the second. This compounding effect reduces the \emph{effective leverage ratio} to $3.04\% / 1.76\% = 1.73$, significantly below the intended leverage of $2$. 
	This example highlights how LETFs can deviate from their target leverage over multiple periods due to compounding, even when the underlying index experiences modest fluctuations. \hfill \qedsymbol
\end{example}

\medskip
Below, we define the \emph{compounding effect} as the difference between the cumulative return of an LETF and the corresponding multiple of the underlying ETF’s cumulative return. Unlike prior studies that focus solely on negative deviations and label them as {volatility drag}, see \cite{avellaneda2010path,trainor2013forecasting}, we adopt a broader perspective by recognizing that the compounding effect can be either positive or negative, depending on market conditions.

\medskip
\begin{definition}[Compounding Effect] \rm \label{definition: Compounding Effect}
	Let $\beta \in \mathbb{Z} \setminus \{0\}$ denote the leverage ratio, and let $X_{t}^{\rm LETF}$ and~$X_{t}^{\rm ETF}$ represent the returns at time $t$ for the LETF and the corresponding benchmark ETF, respectively. Then, the \emph{compounding effect} over an $n$-period horizon, call it $\mathtt{CE}_n$, is defined as the difference between the cumulative return of the LETF and $\beta$ times the cumulative return of the~ETF:
	\begin{align} \label{eq: volatility equation}
		\mathtt{CE}_n :=	R_{n}^{\rm LETF, \beta} - \beta R_{n}^{\rm ETF},
	\end{align}
	where $R_{n}^{\rm LETF, \beta} =  \prod_{t = 1}^{n}(1 + X_{t}^{\rm LETF}) - 1  $ and $R_{n}^{\rm ETF} = \prod_{t = 1}^{n}( 1 + X_{t}^{\rm ETF}) - 1$ represent the $n$-period cumulative returns of the LETF with leverage ratio $\beta$ and the underlying benchmark ETF,~respectively. 
\end{definition}

\medskip
\begin{remark} \rm
	A positive compounding effect implies that the LETF outperforms its target leverage multiple, while a negative compounding effect implies underperformance. The returns in Definition~\ref{definition: Compounding Effect} can represent daily, weekly, monthly, or other periodic returns, depending on the rebalancing frequency. Common values for the leverage ratio in practice include~$\beta \in \{-2, -1, 2, 3\}$, see also Sections~\ref{section: Compounding Effect Analysis} and \ref{section: Empirical Studies}. When $\beta$ is negative, these LETFs are called \emph{inverse} ETFs. 
\end{remark}

\subsection{Fees and Tracking Error Considerations}  

Real-world LETFs incur annual expense fees,  which can significantly impact performance over time. 
ETF fees are typically much smaller than LETF fees and tracking errors. Although ETF fees exist, they typically contribute only a marginal correction relative to the dominant effects in LETF performance. Henceforth, we take the ETF’s daily return to be  
$
X_t^{\rm ETF} =  X_t.
$
The LETF’s daily return is modified to include \textit{fees} and \textit{tracking errors}:
\[
X_t^{\rm LETF} =  \beta X_t - f + e_t,
\]
where $\beta$ is the leverage ratio, $f$ is the daily cost (incorporating expense ratio and transaction costs, e.g., $ f \approx \tfrac{1}{252}( f_{{\rm expense}} + f_{{\rm transaction}}) $,  and  $e_t$ is a zero-mean tracking error with variance \(\tau^2\).

\medskip
\begin{theorem}[Expected Compounding Effect Approximation] \label{theorem: Compounding Effect Derivation} 
	Let $\beta \in \mathbb{Z} \setminus \{0\}$. Assume that the underlying daily returns $\{X_t\}_{t \geq 1}$ is strictly stationary with mean $\mu = \mathbb{E}[X_t]$ and autocorrelation $\gamma_k = \mathbb{E}[X_t X_{t+k}]$. Let $f$ denote the daily fee, and $\{e_t\}$ be the tracking error and $|X_t|, |f|, |e_t| \leq M$, for some sufficiently small $M \ll 1$ with $\mathbb{E}[e_t] = \mathbb{E}[e_t X_s] = \mathbb{E}[e_t e_s] = 0$ for $t\neq s$. Then, the compounding effect over $n$ days~satisfies
	\[
	\mathbb{E}[ \mathtt{CE}_n] = -n f + \sum_{k=1}^{n-1} (n-k) \left( \beta(\beta - 1) \gamma_k - 2\beta f \mu \right) + \binom{n}{2} f^2 + O(n^3 M^3).
	\]
\end{theorem}

\begin{proof}
	See Appendix~\ref{appendix: technical proofs in Preliminaries}.
\end{proof}

\begin{remark} \rm
	$(i)$ Theorem~\ref{theorem: Compounding Effect Derivation} shows that the expected compounding effect depends on autocorrelation $\mathbb{E}[X_t X_s] = \gamma_{|t-s|}$ rather than just volatility. If there are no fees and no tracking error, i.e.,~$f = e_t = 0$, then the first-order term vanishes. We obtain $\mathbb{E}[ \mathtt{CE}_n] \approx \sum_{k=1}^{n-1} (n-k)   \beta(\beta - 1) \gamma_k$,  showing that
	\emph{autocorrelation, not volatility, drives expected compounding}; see Lemmas~\ref{lemma: Positive Expected Compounding Effect} and~\ref{lemma: n-period case for serially correlated AR(1) returns}. Moreover, under the i.i.d. model $\gamma_k=0$ for $k \ge 1$, hence
	$\mathbb{E}[ \mathtt{CE}_n] = -nf$, which is strictly negative only because of costs.
	$(ii)$ If there are nonzero fees \(f > 0\) and nonzero tracking error \(e_t\), then the fees contribute a first-order negative term \(-nf\), and the second-order correction further adjusts the compounding effect. For realistic (small) values of \(f\) and \(\tau\), these corrections will dampen the compounding effect or even reverse its sign if costs are sufficiently~high.
\end{remark}

\section{Compounding Effect Analysis: A Revisit}\label{section: Compounding Effect Analysis}
This section presents a comprehensive analysis of the compounding effect, examining how changes in key parameters---leverage ratio, volatility, and rebalancing frequency---affect LETF performance. By exploring these factors across independent, serially correlated return scenarios, as well as continuous-time regime switching models, we aim to provide actionable insights for managing LETFs under varying market~conditions. To isolate the compounding effect, we henceforth assume there are no fees and no tracking~error.

\subsection{Serially Independent Returns Case: A Discrete-Time Benchmark}\label{math induction}
We first analyze the compounding effect under the assumption of serially independent returns, representing a baseline scenario in an idealized market. 
Specifically, we consider the case where the returns of the underlying ETF, $X_{t}^{\rm ETF}$,  are independent, with a common mean $\mu \in \mathbb{R}$ and variance~$\sigma^{2} > 0$. The returns of the LETF, $X_{t}^{\rm LETF}$, are assumed to perfectly track the underlying ETF with a leverage ratio $\beta \in \mathbb{Z} \setminus \{0\}$ with no fees. 
Notably, due to the limited liability for investors regarding ETFs and LETFs, we have $\beta\mu > -1$ in practice, i.e., the LETF does not collapse. Under these circumstances, the following lemma indicates that the expected compounding effect is always nonnegative, regardless of the leverage ratio or the holding period.

\medskip
\begin{lemma}[Nonnegative Expected Compounding Effect] \label{lemma: Positive Expected Compounding Effect}
	Let $\{X_t\}_{t\geq 1}$ be a sequence of independent returns with mean $\mathbb{E}[X_t] = \mu$ and let $\beta \in \mathbb{Z} \setminus \{0, 1\}$ be the leverage ratio with $\beta \mu > -1$. Then, the expected compounding effect over an $n$-period horizon satisfies
	\begin{align} \label{eq: Positive Expected Compounding Effect}
		\mathbb{E}\left[ \mathtt{CE}_n \right] 
		& = [ (1 + \beta \mu)^{n}-1] -\beta[(1 + \mu)^{n}-1] \geq 0\;  \text{ for all } n \geq 1.
	\end{align}
	Moreover, if $\mu = 0$, then $\mathbb{E}\left[ \mathtt{CE}_n \right]  = 0$ for all $n \geq 1$. 
\end{lemma}

\begin{proof}
	See Appendix~\ref{appendix: technical proofs}.
\end{proof}

\begin{remark}\rm
	Lemma~\ref{lemma: Positive Expected Compounding Effect} implies that for independent returns, LETFs tend to outperform expectations regardless of the leverage ratio and holding period.
	This is also evident from the subsequent simulation results in Section~\ref{section: Empirical Studies}.
	If leverage is not allowed, i.e., $\beta = 1$, then $ \mathtt{CE}_n  = 0$ with probability one. Lastly, while the lemma proves this result for integer values of the leverage ratio $\beta$, in theory, it can take any real value; however,  practical trading restricts $\beta$ to~integers.
\end{remark}

\subsubsection{Monte Carlo Simulations with Independent Return Model.} \label{subsubsection: Monte Carlo Simulations with Independent Return Model}
To validate our result, we use daily adjusted closing price data for \texttt{SPY}, an ETF that tracks the S\&P 500 index, spanning from January 1993 to December 2023 with about a total of 7,800 trading days.  \texttt{SPY} is chosen for its broad market representation and high liquidity, making it an ideal candidate for evaluating LETF performance.  Using this data, we estimate the daily returns and their standard deviation, which is approximately 1\%.

\medskip
\emph{Compounding Effect with Various Leverage Ratios.}
To understand how leverage ratios $\beta$ affect the expected compounding effect $\mathbb{E}[\mathtt{CE}_n]$, we assume that the daily returns are normally distributed, with an annualized mean ranging from -20\% to 20\%. 
For each simulation, we independently draw daily return samples from this distribution to generate 100,000 return paths over one year across various combinations.\footnote{We consider one unit of leveraged ETF paired with $\beta$ (its corresponding leverage ratio) unit of ETF as a combination or a portfolio. Here, we consider $\beta \in \{-2, -1, 2, 3\}$, which is commonly found in the market.}

The theoretical and simulated results, shown in Figure~\ref{Fig: simulated and theoretical compounding effect with different leverage multiplier}, demonstrate that LETFs generally have a positive expected compounding effect across various leverage levels, indicating that LETFs generally outperform their underlying target on average. 
From the figure, we also see that the compounding effect diminishes as expected returns $\mu$ approach zero, and it becomes more pronounced when it deviates significantly from zero. These findings are consistent with~Lemma~\ref{lemma: Positive Expected Compounding Effect}, which predicts a positive expected compounding effect for independent~returns.
Moreover, consistent with existing literature, e.g., \cite{trainor2013forecasting, abdou2017accounting}, we see that LETFs with higher leverage ratios exhibit a stronger compounding effect.

\begin{figure}[htbp]
	\centering
	\includegraphics[width = 0.6\textwidth]{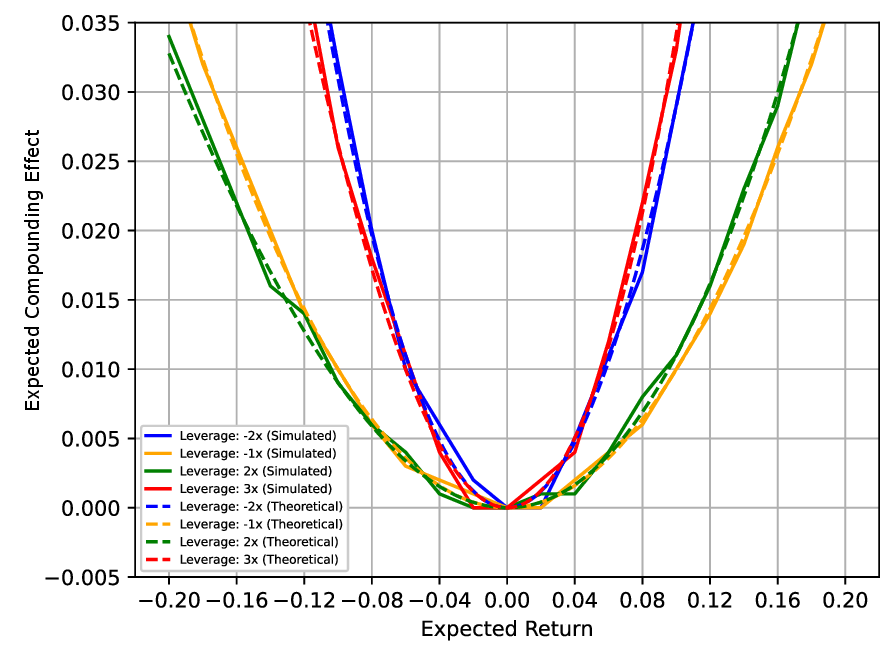}
	\caption{Simulated and Theoretical Compounding Effect $\mathbb{E}[\mathtt{CE}_n]$ with Different Leverage Ratios (Independent Return Case) where Theoretical $\mathbb{E}[\mathtt{CE}_n]$ is computed by Lemma~\ref{lemma: Positive Expected Compounding Effect} and Simulated $\mathbb{E}[\mathtt{CE}_n]$ is obtained via Monte Carlo simulations. A positive effect is observed across different values of $\beta$ and $\mu$.}
	\label{Fig: simulated and theoretical compounding effect with different leverage multiplier}
\end{figure}

\medskip
\emph{Compounding Effect with Different Levels of Volatility.}
After examining the impact of different leverage levels $\beta$ on the compounding effect, we investigate how varying volatility levels $\sigma$ influence the expected compounding effect $\mathbb{E}[\mathtt{CE}_n]$. Specifically, we simulate scenarios with daily volatility levels of 0.5\%, 1\%, and 1.5\%, where the~1\% level corresponds to the volatility estimate derived from \texttt{SPY} return data mentioned in Section~\ref{subsubsection: Monte Carlo Simulations with Independent Return Model}.   The other two levels represent adjustments of~$0.5$ standard deviations above and below this baseline. 

The results, shown in Figure~\ref{fig: compounding effect with different volatility}, indicate that when volatility changes, the expected compounding effect remains relatively stable. This suggests that, for ETFs with similar expected returns, the expected compounding effect exhibits consistent behavior despite variations in volatility.\footnote{Here, we only present the results when the leverage ratio $\beta = 2$ since similar patterns are shown across different leverage combinations, presenting a similar shape of their respective curves.} However, as we shall see later in the case of serially correlated returns,  the volatility level would impact the compounding effect.

\begin{figure}[htbp]
	\centering
	\includegraphics[width=0.65\textwidth]{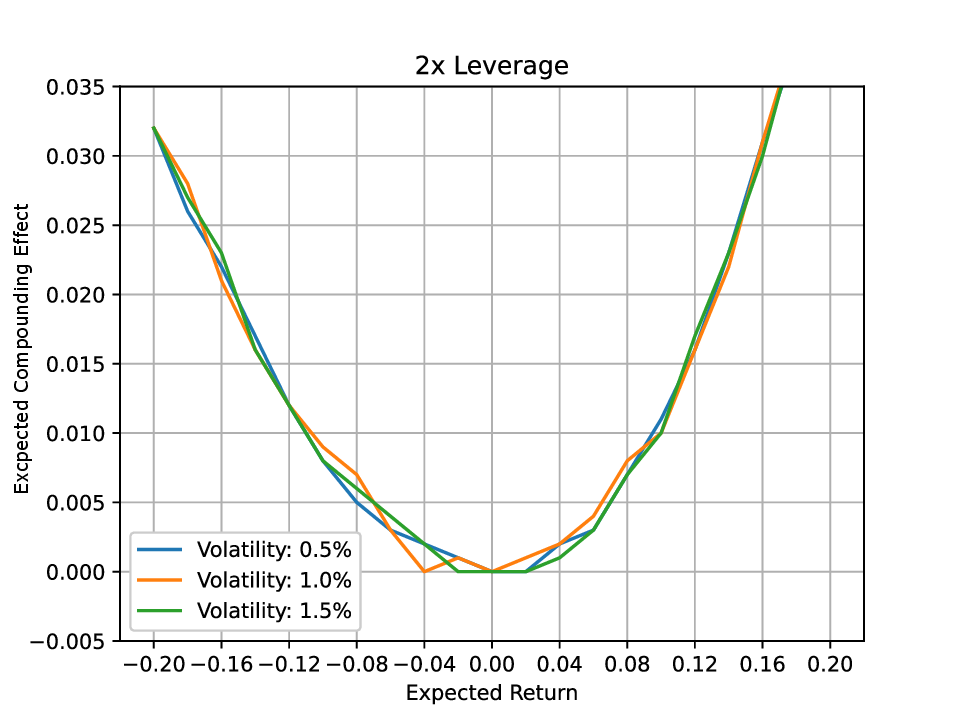}
	\caption{Compounding Effect with Different Volatility (Independent Return Case).}
	\label{fig: compounding effect with different volatility}
\end{figure}

\begin{figure}[htbp]
	\centering
	\includegraphics[width=0.65\textwidth]{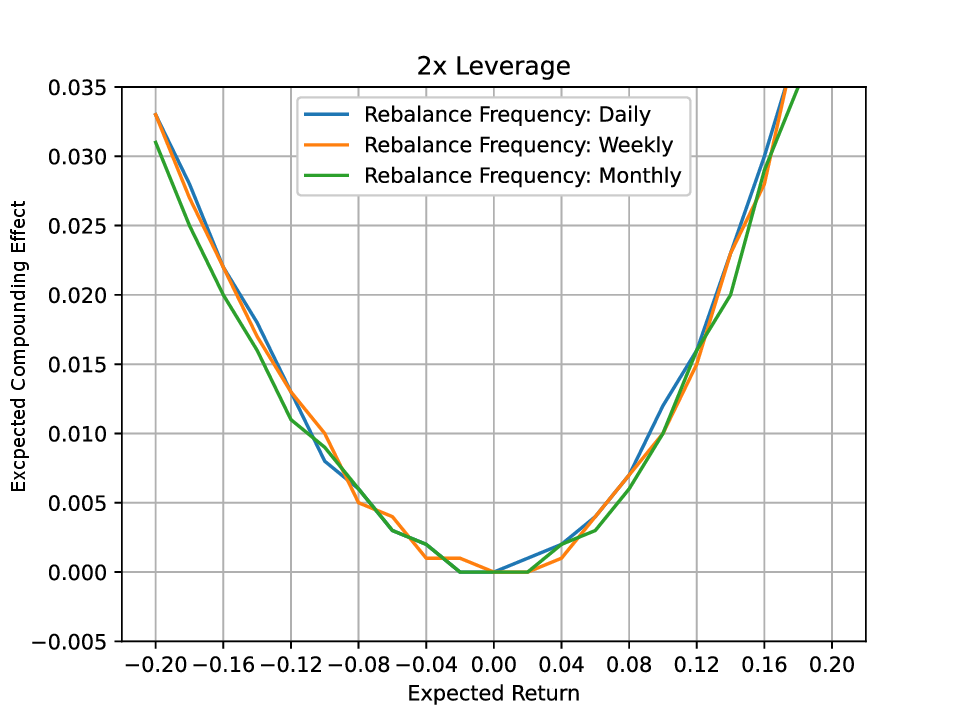}
	\caption{Compounding Effect with Different Rebalancing Frequencies (Independent Return Case).}
	\label{fig: compounding effect with different rebalancing frequency}
\end{figure}

\medskip
\emph{Rebalancing Frequency Effects.}
Next, we analyze how the rebalancing frequency of LETFs influences the expected compounding effect by conducting simulations with daily, weekly, and monthly rebalancing intervals. Using a daily volatility level of 1\%, we compare the results across different rebalancing frequencies. The results in Figure~\ref{fig: compounding effect with different rebalancing frequency} demonstrate that changing the rebalancing frequency does not affect the compounding effect when returns are independent. This result suggests that, in the absence of serial correlation, rebalancing frequency has minimal impact on LETF performance. Consequently, under such conditions, investors may choose less frequent rebalancing strategies to reduce transaction costs without sacrificing performance.

\subsubsection{Portfolio Construction Implications: Independent Return Model.} \label{subsubsection: Portfolio Construction Implications}

Under the frictionless independent returns benchmark, we can construct arbitrage-like portfolios by combining long positions in LETFs with short positions in the corresponding ETFs. For leverage $\beta > 1$, a portfolio with one unit of LETF and $\beta$ units of short position in the underlying ETF captures the positive expected compounding effect $\mathbb{E}[\mathtt{CE}_n]$. For $\beta < 0$, we substitute short ETF positions with long positions in inverse ETFs.
The return of such a portfolio is 
$$
R_{\rm portfolio} := R_n^{\rm LETF} - \beta R_n^{\rm ETF} = \mathtt{CE}_n
$$
and the expected return is simply $\mathbb{E}[R_{\rm portfolio}] = \mathbb{E}[\mathtt{CE}_n] > 0$ under independence and $\mu \neq 0$, as established in Lemma~\ref{lemma: Positive Expected Compounding Effect}.
This theoretical property offers potential arbitrage opportunities in markets where returns approximate independence. However, as shown in subsequent sections, return autocorrelation significantly alters compounding behavior, and thus the profitability of such strategies is highly regime-dependent. See also~\cite{rappoport2018arbitrage} for a empirical evidence that ETF arbitrage may be less effective under illiquid condition.

\subsection{Serially Correlated Returns}\label{Analysis of Serially Correlated Returns}
We now extend our analysis to the case of serially correlated returns, where market dynamics exhibit short-term trends or reversals. To capture the temporal dependence, we begin with an autoregressive process of order 1 (AR(1)) model.

\subsubsection{Returns with AR(1) Model.}
Specifically, we model the returns $X_t$ using an AR(1) process: $X_t = c + \phi X_{t-1} + \varepsilon_t $
with zero intercept $c=0$, mean-zero innovation $\mathbb{E}[\varepsilon_t] = 0$, and constant variance ${\rm var}(\varepsilon_t) = \sigma^2>0$, and $\varepsilon_t$ are independent of $X_t$. Moreover, we assume $\mathbb{E}[X_t] = 0$ and~$\mathbb{E}[X_t X_{t-1}] = \phi \frac{\sigma^2}{1-\phi^2}$, where $|\phi| <1$. 
This setup reflects scenarios where returns exhibit serial correlation, either positive (momentum) or negative (mean reversion), which can significantly alter the compounding effect compared to the independent return case. The following lemma indicates that the expected compounding effect depends on the sign of the autocorrelation coefficient $\phi$.

\medskip
\begin{lemma}[Compounding Effect for AR(1) Return Model] \label{lemma: n-period case for serially correlated AR(1) returns}  
	Let $\beta \in \mathbb{Z}\setminus\{0, 1\}$ be the leverage ratio, and let $\{X_t\}_{t\geq 1}$ follow an AR(1) model. Suppose there exists a constant $m \in (0, 1)$ such that $\mathbb{P}(|X_t| < m ) = 1$ for all $t =1 \dots, n$.
	Then, the following statements hold:
	\begin{itemize}
		\item[(i)] If $\phi>0$, then $\mathbb{E}[ \mathtt{CE}_n] > 0$.
		\item[(ii)] If $\phi < 0$, then $\mathbb{E}[ \mathtt{CE}_n] < 0$.
	\end{itemize}
\end{lemma}

\begin{proof}
	See Appendix~\ref{appendix: compounding effect for serially correlated returns}.
\end{proof}

\begin{remark} \rm
	Lemma~\ref{lemma: n-period case for serially correlated AR(1) returns} establishes that short-term serial dependence---specifically, positive or negative autocorrelation---significantly influences LETF performance. When returns exhibit positive autocorrelation, positive correlation amplifies returns aligned with the prevailing trend, thereby amplifying gains. In contrast, negative autocorrelation induces a mean-reverting pattern that misaligns rebalancing with price movements, leading to a performance drag due to buying high and selling low.
	
	While the AR(1) model suggests that $\phi > 0$ leads to LETF outperformance and $\phi < 0$ leads to underperformance, this relationship weakens in the presence of volatility clustering. As we will see in Section~\ref{subsubsection: Returns with AR-GARCH Models}, the AR(1)-GARCH(1,1) model reveals that volatility persistence dominates, reducing the predictive power of $\phi$ in real-world markets.
\end{remark}

\subsubsection{Monte Carlo Simulations with AR(1) Model.} \label{subsection: Monte Carlo Simulations with AR(1) Model}
In this section, we simulate how LETF performance varies with different values of AR coefficient $\phi$, using Monte Carlo simulations. Specifically, we generate $10,000$ sample paths of daily returns over a one-year horizon and examine how the expected compounding effect $\mathbb{E}[\mathtt{CE}_n]$ varies across a range of AR coefficients with $ \phi \in [-0.9, 0.9]$.

\medskip
\emph{Compounding Effect with Various Leverage Ratios.}
Figure~\ref{fig: simulated and theoretical compounding effect with different leverage in AR(1) returns} demonstrates the impact of different AR coefficients $\phi$ on the expected compounding effect $\mathbb{E}[ \mathtt{CE}_n]$ across various combinations of LETFs. Unlike the independent return case, where the expected compounding effect is positive, the AR(1) structure induces sign dependence: 
For $\phi > 0$ (momentum-driven markets), the expected compounding effect is positive,  enhancing LETF performance. Conversely, the expected compounding effect becomes negative for $\phi < 0$ (mean-reverting markets), indicating performance drag due to return reversals. 
This observation highlights the critical role of return dynamics in determining LETF performance and the necessity of accounting for temporal correlation when evaluating LETF~strategies.

\begin{figure}[htbp]
	\centering
	\includegraphics[width=0.55\textwidth]{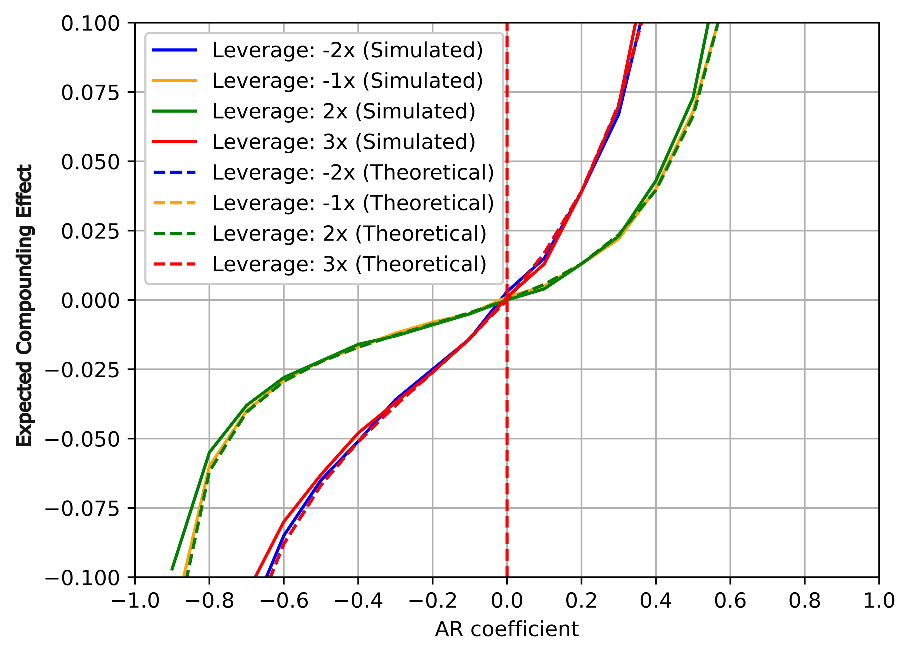}
	\caption{Simulated and Theoretical Compounding Effect with Different Leverage Ratio $\beta$. (AR(1) Returns  where theoretical $\mathbb{E}[\mathtt{CE}_n]$ is computed by Lemma~\ref{lemma: Positive Expected Compounding Effect} and simulated $\mathbb{E}[\mathtt{CE}_n]$ is obtained via Monte Carlo simulations with $10,000$ sample paths for each $\phi$.)}
	\label{fig: simulated and theoretical compounding effect with different leverage in AR(1) returns}
\end{figure}

\begin{figure}[htbp]
	\centering
	\includegraphics[width=0.85\textwidth]{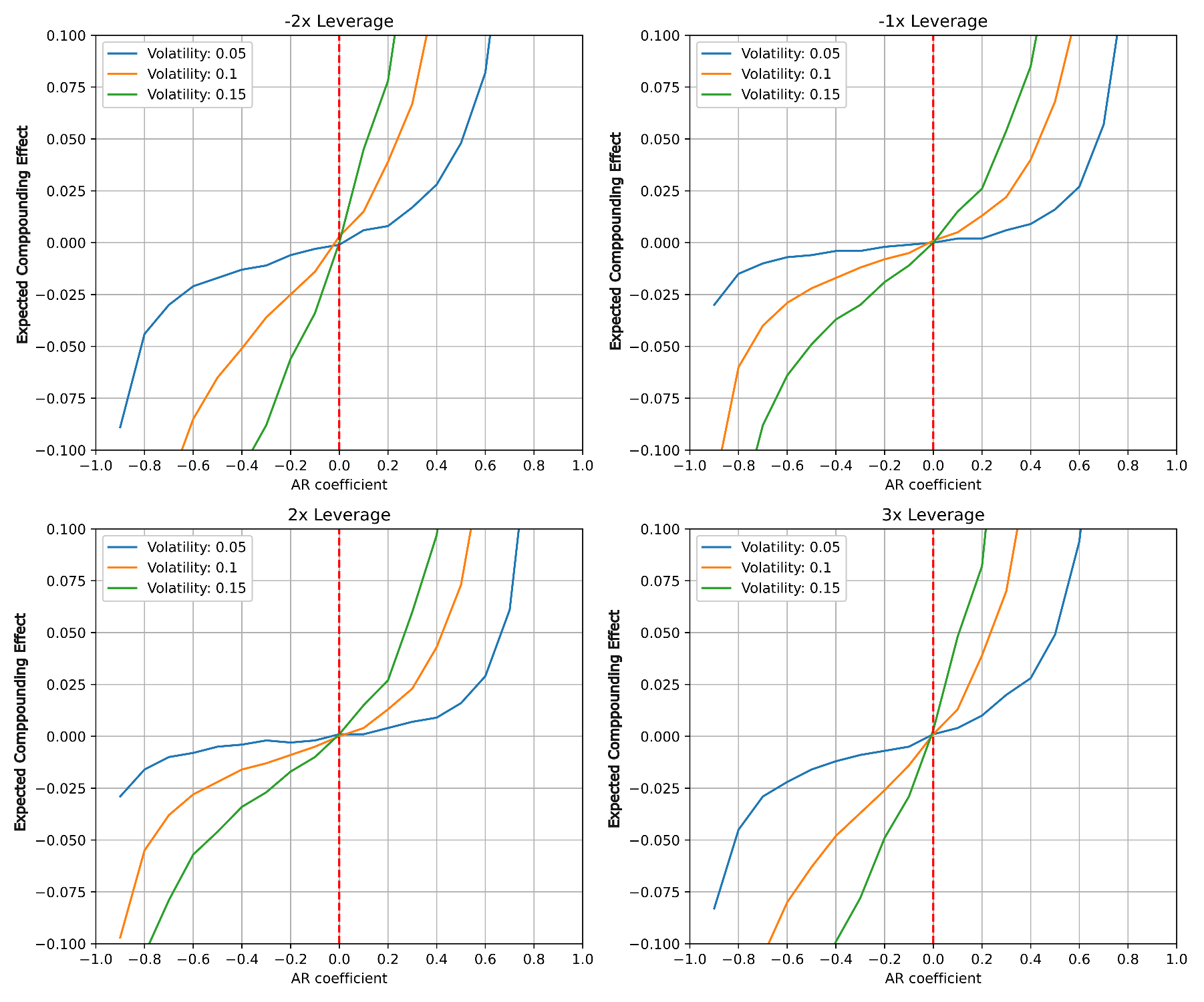}
	\caption{Expected Compounding Effect with Different Volatility (AR(1) Returns).}
	\label{fig: compounding effect with different volatility in AR(1) returns}
\end{figure}

\medskip
\emph{Compounding Effect with Different Levels of Volatility.}
We examine how volatility modulates the expected compounding effect in the presence of serial correlation among returns. As shown in Figure~\ref{fig: compounding effect with different volatility in AR(1) returns}, increased volatility amplifies the expected compounding effect when the AR coefficient~$\phi >0$, reflecting enhanced return potential under momentum-driven market conditions. In such scenarios, LETFs benefit from larger price movements, as rebalancing adjustments amplify returns in the direction of the trend. However, when the positive correlation weakens, increased volatility introduces greater randomness, reducing the effectiveness of rebalancing.

Conversely, when $\phi < 0$, corresponding to mean-reverting markets, higher volatility exacerbates return fluctuations and induces performance drag. Rebalancing in this regime amplifies exposure during losses and suppresses it during recoveries, generating a pronounced negative compounding effect. These findings highlight the dual role of volatility: while it can enhance LETF performance in trending markets, it poses significant risks in mean-reverting environments.

\medskip
\emph{Compounding Effect with Varying Rebalancing Frequencies.}
We examine how rebalancing frequency impacts the compounding effect under AR(1) return model. We simulate daily, weekly, and monthly rebalancing intervals. The results, presented in Figure~\ref{fig: compounding effect with different rebalancing frequency in AR(1) model}, indicate that when the AR coefficient $\phi$ is positive (indicating momentum-driven markets), portfolios with more frequent rebalancing exhibit a stronger positive compounding effect. This outcome shows the importance of timely rebalancing in capturing upward trends, as daily rebalancing allows LETFs to efficiently capitalize on market~movements.

In contrast, when the AR coefficient is negative (mean-reverting markets), frequent rebalancing becomes detrimental, leading to a cycle of buying high and selling low. Interestingly, extending the rebalancing interval to a week reduces the negative compounding effect, while monthly rebalancing nearly eliminates it. These results suggest that in oscillating markets, less frequent rebalancing can mitigate performance drag, providing a strategic advantage to investors who choose LETFs with longer rebalancing intervals.

\begin{figure}[htbp]
	\centering
	\includegraphics[width=0.9\textwidth]{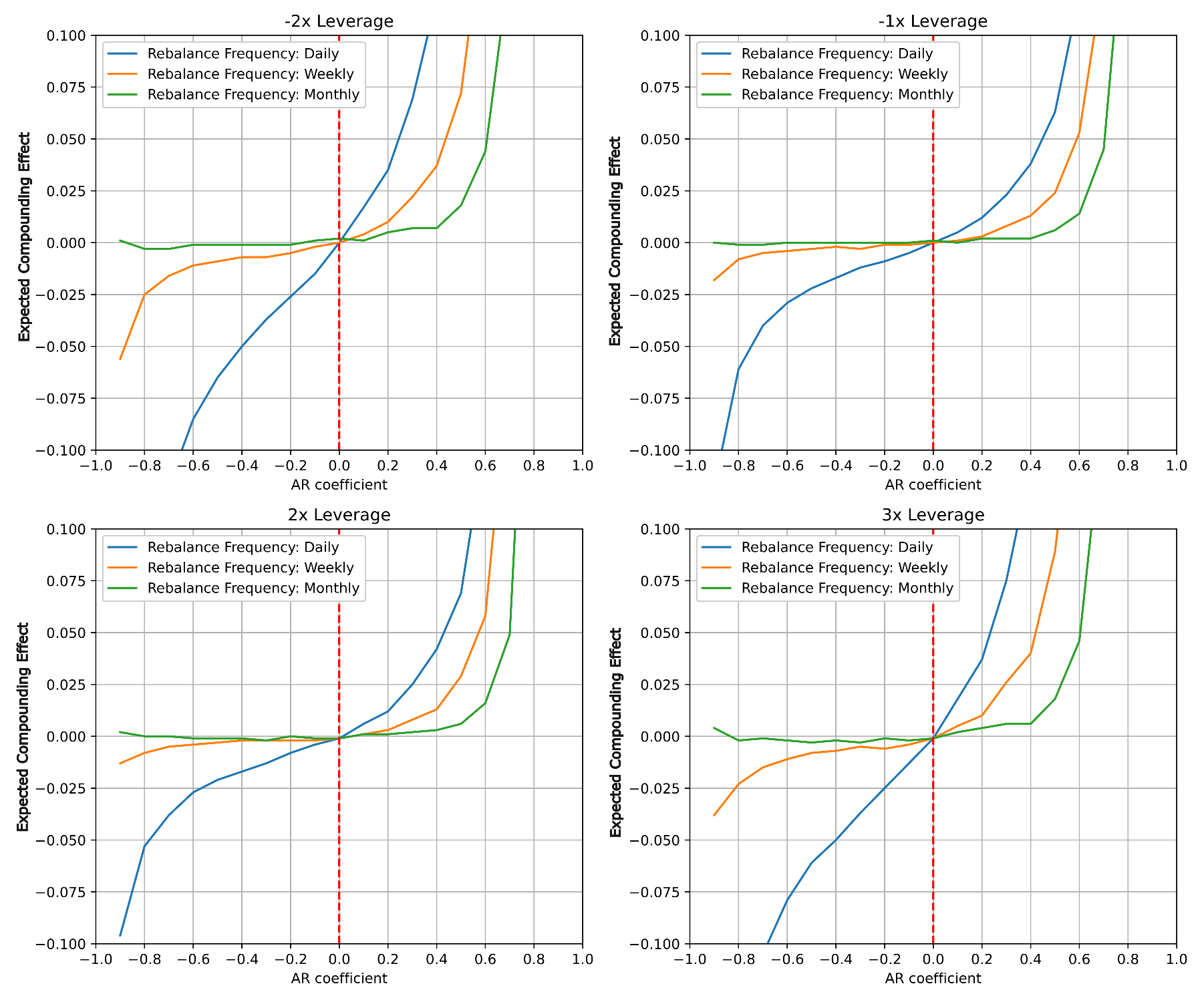}
	\caption{Expected Compounding Effect with Different Rebalancing Frequencies (AR(1) Returns.)}
	\label{fig: compounding effect with different rebalancing frequency in AR(1) model}
\end{figure}

\subsubsection{Portfolio Construction Implications: AR(1) Model.}
Under AR(1) dynamics, arbitrage-like constructions, as seen in Section~\ref{subsubsection: Portfolio Construction Implications}, analogous to the independent case may persist when the AR coefficient 
$\phi >0$ (momentum regime), as established in Lemma~\ref{lemma: n-period case for serially correlated AR(1) returns}. However, when  $\phi <0$ such portfolios exhibit negative expected compounding, making the strategy counterproductive.

\subsubsection{Returns with AR-GARCH Models.} \label{subsubsection: Returns with AR-GARCH Models}
We now extend our analysis from synthetic return models to an empirically calibrated AR-GARCH model. The preceding sections considered independent returns and AR(1) returns with controlled parameters to isolate key drivers of the compounding effect. Here, we adopt a more realistic framework by estimating model parameters directly from historical data.

Rather than conducting parameter sweeps, we estimate parameters via the maximum likelihood method using daily returns for \texttt{SPY} and LETFs, then assess their implications for the compounding effect through Monte Carlo simulation.  Let $X_t$ represent daily returns. A general ARMA($p, q$)-GARCH($r,s$) structure is given by:
\begin{align*}
	X_t &= \mu + \sum_{i=1}^{p}\phi_i X_{t-i} + \sum_{j=1}^{q}\theta_j\epsilon_{t-j} + \epsilon_t,\\
	\epsilon_t &= \sigma_t z_t, \quad z_t \sim \mathcal{N}(0,1),\\
	\sigma_t^2 &= \omega + \sum_{i=1}^{r}\alpha_i \epsilon_{t-i}^2 + \sum_{j=1}^{s}\beta_j\sigma_{t-j}^2.
\end{align*}
For tractability and empirical relevance, we adopt the AR(1)-GARCH(1,1) specification:
\begin{align*}
	X_t &= \mu + \phi X_{t-1} + \epsilon_t,\\
	\sigma_t^2 &= \omega + \alpha \epsilon_{t-1}^2 + \beta \sigma_{t-1}^2.
\end{align*}

Our empirical results reveal that, under AR(1)-GARCH(1,1), the dominant driver of LETF performance is volatility persistence, as captured by the sum $\alpha + \beta$, rather than return autocorrelation. While the AR(1) coefficient $\phi$ remains statistically significant, its impact on LETF performance is diminished by the presence of volatility clustering.

Indeed, we fit the AR(1)-GARCH(1,1) model to historical daily returns for \texttt{SPY} and several LETFs (\texttt{SSO, SPXL, SDS, SH}) using maximum likelihood estimation (MLE), see \cite{casella2024statistical}. The data ranges from February 2010 to December 2023. The estimation results, reported in Table~\ref{table: garch_results_full}, indicate mild mean-reversion (negative AR(1) coefficients) and strong volatility clustering, with significant GARCH parameters.

\begin{table}[htbp]
	\small
	\centering
	\caption{AR(1)-GARCH(1,1) Parameter Estimates for ETFs and LETFs}
	\label{table: garch_results_full}
	\begin{tabular}{lccccc}
		\toprule
		ETF & $\mu$ (Const.) & AR(1) $\phi$ & $\omega$ & $\alpha$ & $\beta$ \\
		\midrule
		\texttt{SPY} & 0.0918$^{***}$ (0.0130) & -0.0490$^{***}$ (0.0188) & 0.0357$^{***}$ (0.0071) & 0.1747$^{***}$ (0.0221) & 0.7969$^{***}$ (0.0212) \\
		\texttt{SSO} & 0.1674$^{***}$ (0.0259) & -0.0433$^{**}$ (0.0191) & 0.1434$^{***}$ (0.0286) & 0.1786$^{***}$ (0.0223) & 0.7939$^{***}$ (0.0211) \\
		\texttt{SPXL} & 0.2413$^{***}$ (0.0389) & -0.0384$^{**}$ (0.0191) & 0.3268$^{***}$ (0.0634) & 0.1805$^{***}$ (0.0220) & 0.7923$^{***}$ (0.0208) \\
		\texttt{SDS} & -0.1973$^{***}$ (0.0260) & -0.0483$^{***}$ (0.0187) & 0.1338$^{***}$ (0.0280) & 0.1760$^{***}$ (0.0212) & 0.7976$^{***}$ (0.0209) \\
		\texttt{SH} & -0.0975$^{***}$ (0.0130) & -0.0487$^{***}$ (0.0187) & 0.0340$^{***}$ (0.0069) & 0.1770$^{***}$ (0.0220) & 0.7963$^{***}$ (0.0207) \\
		\bottomrule
		\multicolumn{6}{l}{\small Standard errors in parentheses; significance: $^{***}p<0.01$ and $^{**} p <0.05.$}\\
	\end{tabular}
\end{table}

\textit{Monte Carlo Simulation and Compounding Effect}. We use the estimated AR(1)-GARCH(1,1) model parameters to simulate 10,000 return paths for \texttt{SPY} and its leveraged counterparts over a one-year horizon (252 trading days). From these simulations, we rigorously quantify the compounding~effect:
\begin{equation}
	\mathtt{CE}_{252} = \left(\prod_{t=1}^{252}(1+\beta X_t)-1\right) - \beta \left(\prod_{t=1}^{252}(1+X_t)-1\right),
\end{equation}
where $\beta$ denotes the ETF leverage ratio and $X_t$ denotes simulated daily returns of \texttt{SPY}. The results from these comprehensive simulations are summarized in Table~\ref{table:volatility_effect_simulation}.
Specifically, the table demonstrates significant compounding effects that become larger as leverage increases. We see that volatility clustering affects the performance of LETFs, producing consistent deviations from the stated leverage multiples, confirming the importance of explicitly modeling volatility and serial dependence in the model.

\begin{table}[htbp]
	\centering
	\caption{Simulated Compounding Effects Using AR(1)-GARCH(1,1) for \texttt{SPY} and Its Associated LETFs}
	\label{table:volatility_effect_simulation}
	\begin{tabular}{lcc}
		\toprule
		LETF Ticker ($\beta$) & Expected Compounding Effect & Standard Deviation \\
		\midrule
		\texttt{SSO} (2x) & \text{0.0744} & \text{0.1844} \\
		\texttt{SPXL} (3x) & \text{0.2443} & \text{0.6383} \\
		\texttt{SDS} (-2x) & \text{0.1664} & \text{0.2756} \\
		\texttt{SH} (-1x) & \text{0.0597} & \text{0.1063} \\
		\bottomrule
	\end{tabular}
\end{table}

\medskip
\textit{Autocorrelation Versus Volatility Clustering.} \rm
While the AR(1) coefficients estimated in Table~\ref{table: garch_results_full} are statistically significant, their magnitudes are modest compared to the strong persistence exhibited in the GARCH terms (i.e., high values of $\alpha + \beta$). This suggests that, in empirical return dynamics, volatility clustering exerts a more dominant influence on compounding effects than return autocorrelation alone. 

However, this does not contradict our theoretical findings. Instead, it shows that the relative impact of autocorrelation versus volatility varies by market regime and time scale. Specifically, in low-volatility environments, autocorrelation dominates, see Section~\ref{subsection: Monte Carlo Simulations with AR(1) Model}, whereas in high-volatility periods with clustering, compounding is more sensitive to volatility shocks. 

This interplay implies that the predictive power of autocorrelation diminishes in the presence of highly persistent volatility, as shown in our simulations under the AR(1)-GARCH(1,1) framework.
A full decomposition of these effects remains an important direction for future work. In summary, autocorrelation effects govern compounding behavior in stable regimes, whereas in highly persistent volatility regimes, volatility clustering may overshadow autocorrelation-induced deviations from LETF targets.

\medskip
\textit{Additional Robustness Test: QQQ and Its Associated LETFs}.
To test the robustness of our findings, we expand our empirical analysis beyond the S\&P 500 ETF (\texttt{SPY}) to the NASDAQ-100 ETF (\texttt{QQQ}) and its leveraged counterparts. Given that \texttt{QQQ} tracks a more volatile technology-focused index, this serves as a stress test for our theoretical framework. We apply the same AR(1)-GARCH(1,1) estimation and Monte Carlo simulations to assess whether LETFs on \texttt{QQQ} exhibit similar compounding effects as those based on SPY. The result is summarized in Table~\ref{table: volatility_effect_simulation_for_QQQ}.
The larger compounding effect observed for \texttt{QQQ}-based LETFs compared to SPY-based LETFs can be attributed to \texttt{QQQ}’s historically higher realized volatility. Since volatility drag increases nonlinearly with volatility, this leads to more pronounced deviations in \texttt{QQQ}-based LETFs.

The results above suggest that autocorrelation-driven effects (trend-following benefits or mean-reverting losses) might be weaker in markets where volatility clustering is strong. While our findings hold under AR(1)-GARCH(1,1) dynamics, future work could explore alternative specifications, such as regime-switching or rough volatility models, to further validate the generality of our results.

\begin{table}[htbp]
	\centering
	\caption{Simulated Compounding Effects Using AR(1)-GARCH(1,1) for \texttt{QQQ} and Its Associated LETFs}
	\label{table: volatility_effect_simulation_for_QQQ}
	\begin{tabular}{lcc}
		\toprule
		LETF Ticker ($\beta$) & Mean Compounding Effect & Standard Deviation \\
		\midrule
		\texttt{QLD} (2x) & \text{0.1206} & \text{0.2564} \\
		\texttt{TQQQ} (3x) & \text{0.4027} & \text{0.9647} \\
		\texttt{QID} (-2x) & \text{0.2483} & \text{0.3659} \\
		\texttt{SQQQ} (-3x) & \text{0.4560} & \text{0.6313} \\
		\bottomrule
	\end{tabular}
\end{table}

\subsubsection{Portfolio Construction Implications: AR-GARCH Model.}
In the AR-GARCH setting, the presence of volatility clustering diminishes the predictive power of short-term autocorrelation. While simulations suggest positive compounding effect under estimated parameters, the sign of 
$\mathbb{E}[\mathtt{CE}_n]$ is not guaranteed in theory. Therefore, arbitrage-like portfolios based on compounding asymmetry require careful calibration to prevailing volatility and autocorrelation regimes.

\subsection{Compounding Effect for Continuous-Time Model}
This section examines a continuous-time regime-switching model.
Let $S_t$ be the ETF price process that evolves according to the regime-dependent geometric Brownian motion, see \cite{di1995mean}, i.e., with $S_0 >0$,
$$
dS_t = \mu_{Z_t} S_t dt + \sigma_{Z_t} S_t dW_t
$$
where $\{Z_t\}_{t\geq 0} \in \{1,2,\dots, M\}$ is a continuous-time Markov chain with generator matrix $Q=(Q_{ij})_{1\leq i,j \leq M}$, $Z_t$ governs both the drift~$\mu_{Z_t}$ and volatility $\sigma_{Z_t}$ in each regime, and~$\{W_t\}_{t\geq 0}$ is the Wiener process on a filtered probability space $(\Omega,  \mathcal{F}, \mathbb{P})$ with filtration $\mathcal{F}_t$. We further assume that~$Z_{(\cdot)}$ and $W_{(\cdot)}$ are independent. With $S_0 >0$, the solution of $S_t$ is given by
\begin{align} \label{eq: GBM S_t_regime_switching}
	S_t = S_0 \exp\left( \int_0^t \left( \mu_{Z_s}  - \frac{1}{2}   \sigma_{Z_s}^2 \right) ds + \int_0^t \sigma_{Z_s} dW_s\right).
\end{align}

For an LETF with leverage ratio $ \beta \in \mathbb{Z} \setminus \{0\} $, the corresponding price dynamics, denoted by~$L_t$, can be derived from the underlying ETF price $S_t$ as:
\begin{align} \label{eq: LETF dynamics_regime_switching}
	\frac{d L_t }{L_t} = \beta \frac{dS_t}{S_t}   - fdt.
\end{align}
This formulation reflects how LETFs magnify the returns of their underlying ETFs by a fixed leverage factor $\beta$.
Additionally, the  LETF price dynamics~\eqref{eq: LETF dynamics_regime_switching} with leverage $ \beta \in \mathbb{Z} \setminus \{0\}  $ has
solution
\begin{align} \label{eq: LETF L_t_regime_switching}
	L_t =L_0 \exp\left( \int_0^t \left(  \beta \mu_{Z_s} - \frac{1}{2}  \beta^2 \sigma_{Z_s}^2 \right) ds - ft +  \beta \int_0^t \sigma_{Z_s} dW_s \right).
\end{align}
The derivation of Equation~\eqref{eq: LETF L_t_regime_switching} is relegated to Appendix~\ref{appendix: Compounding Effect for Continuous-Time Return Model}.

\medskip
\begin{remark} \rm
	$(i)$	The source of autocorrelation is the persistence of the Markov state $Z_t$: If $\mathbb{P}(Z_{t+1} = Z_t) > 1/2$, then the return process exhibits positive autocorrelation. Negative autocorrelation can be generated by alternating regimes (e.g., mean-reversion via switching between up/down trends).
	$(ii)$ The continuous-time model implicitly assumes instantaneous and frictionless rebalancing, so that leverage is continuously maintained at level $\beta$. This differs from discrete-time implementations where rebalancing occurs on daily or longer intervals. In practice, continuous rebalancing is unachievable; however, the model provides a theoretical upper bound for compounding effects.
\end{remark}

\medskip
\begin{theorem}[No Intrinsic Volatility Drag in Expectation] \label{theorem: expected compounding effect for diffusion prices_regime}
	Let $\beta \in \mathbb{Z} \setminus \{0\}$ denote the leverage ratio, and consider the continuous-time returns of an ETF and a corresponding LETF, denoted by~$R_t^{\rm ETF}$ and $R_t^{\rm LETF, \beta}$, respectively. Then, the expected compounding effect can be approximated by
	\begin{align} \label{eq: expected compounding effect in regime switching}
		\mathbb{E}[\mathtt{CE}_t]
		&\approx	\sum_{j=1}^{M} \pi_j(t) \left[  \exp\left( (\beta \mu_j  - f )t \right) - \beta \exp\left(  \mu_j t \right) \right] + (\beta - 1).
	\end{align}
	where $\pi_j(t) := \mathbb{E}\left[ \frac{1}{t} \int_0^t \mathbb{1}_{\{Z_s = j\}} ds \right]$ for each $j=1,\dots,M$. 
\end{theorem}

\begin{proof}
	See Appendix~\ref{appendix: Compounding Effect for Continuous-Time Return Model}.
\end{proof}

\begin{remark} \rm
	$(i)$ Equation~\eqref{eq: expected compounding effect in regime switching} becomes exact in the limit of constant regime, i.e., when $Z_t  =j $ for all $t \in [0, T]$; see Corollary~\ref{corollary: Nonnegative CE in Single Regime} to follow. More generally, it serves as a first-order approximation assuming that regime switches are relatively infrequent over the horizon $[0,t]$, and that $Z$ and $W$ are independent.
	In this case, the expected occupation measure $\pi_j(t)$ reflects the time-averaged regime weights and provides a tractable means of interpreting the compounding effect.
	$(ii)$ When~$Z_t$ fluctuates frequently, the approximation becomes less accurate. Nevertheless, the representation in~\eqref{eq: expected compounding effect in regime switching} offers a useful analytical insight into how compounding varies with regime dominance.
\end{remark}

\subsubsection{Three-Regimes Case}.
Let $\beta \in \mathbb{Z} \setminus \{0\}$ denote the leverage ratio, $f \ge 0$ the fixed expense fee, and let $\mu_j \in \mathbb{R}$ denote the drift in regime $j \in \{1,2,3\}$, characterized as follows:
\begin{itemize}
	\item Regime $j=1$: Trending Up ($\mu_1 > 0$)
	\item Regime $j=2$: Trending Down ($\mu_2 < 0$)
	\item Regime $j=3$: Oscillating (Mean-Reverting), where $\mu_3 \approx 0$ and $\sigma_3^2 \gg |\mu_3|$
\end{itemize}
In Regime $3$, although $\mu_3 \approx 0$, the volatility $\sigma_3$ is large, leading to substantial pathwise fluctuation.
Let $\pi_j(t) := \mathbb{E}\left[ \frac{1}{t} \int_0^t \mathbb{1}_{\{Z_s = j\}} ds \right]$ denote the occupation measure of regime $j$ over the interval $[0,t]$, and define the regime-specific compounding kernel:
\[
\Phi_j(t; \beta, f) := \exp((\beta \mu_j - f)t) - \beta \exp(\mu_j t).
\]
Then, the expected compounding effect is approximated by:
\[
\mathbb{E}[\mathtt{CE}_t] \approx \sum_{j=1}^3 \pi_j(t) \Phi_j(t; \beta, f) + (\beta - 1).
\]

\begin{theorem}[Sign of Expected CE under Regime Switching]\label{theorem: Sign of Expected CE under Regime Switching}
	Let $\beta \in \mathbb{Z} \setminus \{0\}$ and $f \ge 0$ be fixed, and suppose $(\mu_1, \mu_2, \mu_3)$ are given constants with $\mu_1 > 0$, $\mu_2 < 0$, and $\mu_3 \approx 0$. Then the sign of $\mathbb{E}[\mathtt{CE}_t]$ is determined entirely by the regime occupation vector $\pi(t) = (\pi_1(t), \pi_2(t), \pi_3(t))$. Specifically:
	\begin{itemize}
		\item If $\sum_{j=1}^3 \pi_j(t) \Phi_j(t; \beta, f) > 1 - \beta$, then $\mathbb{E}[\mathtt{CE}_t] > 0$.
		\item If $\sum_{j=1}^3 \pi_j(t) \Phi_j(t; \beta, f) < 1 - \beta$, then $\mathbb{E}[\mathtt{CE}_t] < 0$.
		\item If $\sum_{j=1}^3 \pi_j(t) \Phi_j(t; \beta, f) = 1 - \beta$, then $\mathbb{E}[\mathtt{CE}_t] = 0$.
	\end{itemize}
\end{theorem}

\begin{proof}
	See Appendix~\ref{appendix: Compounding Effect for Continuous-Time Return Model}.
\end{proof}

\begin{corollary}[Nonnegative Compounding Effect in Single Regime] \label{corollary: Nonnegative CE in Single Regime}
	Let $\beta \in \mathbb{Z} \setminus \{0,1\}$ and $Z_t \equiv j$ be constant, i.e., the process remains in a single regime. Then for any $\mu_j \in \mathbb{R}$, $f = 0$, and all $t \ge 0$, we have:
	\[
	\mathbb{E}[\mathtt{CE}_t] = \exp(\beta \mu_j t) - \beta \exp(\mu_j t) + (\beta - 1) \ge 0
	\]
	with equality if and only if $t = 0$ or $\mu_j = 0$ or $\beta = 1$.
\end{corollary}

\begin{proof}
	See Appendix~\ref{appendix: Compounding Effect for Continuous-Time Return Model}.
\end{proof}

\begin{remark}[Compounding Effect Under Regime Switching] \rm
	$(i)$ In a regime-switching model where $\pi_j(t)$ varies over time and no single regime dominates, the $\mathbb{E}[\mathtt{CE}_t]$ may be negative even though each $\Phi_j(t; \beta, f) \ge 0$ in isolation (as shown in the corollary for $f=0$). This occurs because regime mixing destroys the directional coherence necessary for compounding to accumulate effectively.
	$(ii)$ While GARCH models account for conditional heteroskedasticity, regime-switching models capture drift and volatility persistence via state dynamics. Volatility clustering in AR-GARCH is mimicked in our model by allowing high-volatility states with persistent occupation probabilities. Hence, the regime-switching framework can subsume AR-GARCH-like behavior when transition rates are sufficiently sticky.
\end{remark}

\subsubsection{Portfolio Construction Implications: Regime Switching Model/}
In the continuous-time regime-switching model, expected compounding effect is a function of regime occupation weights $\pi_j(t)$. Arbitrage-like portfolios, as seen in Section~\ref{subsubsection: Portfolio Construction Implications}, may yield positive expected returns when the process spends significant time in trending regimes with positive drift. However, extended periods in oscillating or high-volatility regimes may negate the compounding benefit. Thus, strategy effectiveness depends critically on regime persistence


\section{Empirical Studies}\label{section: Empirical Studies}
This section empirically examines the compounding effect of LETFs using historical data from the SPDR S\&P 500 ETF (Ticker: \texttt{SPY}) and the Nasdaq-100 ETF (Ticker: \texttt{QQQ}) under varying market conditions. We aim to validate our theoretical findings that LETFs outperform in trending markets and underperform in mean-reverting markets due to volatility.

\subsection{Hypotheses and Empirical Strategy}

To rigorously test our theoretical predictions, we formulate the following hypotheses:

\textbf{H1 (Trending Markets):} LETFs exhibit a positive compounding effect in strongly trending markets, regardless of direction (e.g., Financial Crisis, Post-Crisis Recovery).

\textbf{H2 (Mean-Reverting Markets):} LETFs exhibit a negative compounding effect in oscillating or mean-reverting markets, as frequent rebalancing amplifies losses (e.g., Sideways Markets).

\textbf{H3 (Volatility Impact):} The compounding effect is stronger for higher leverage ratios (e.g., $\beta =~\pm3$), with larger deviations in high-volatility environments.

To test these hypotheses, we analyze LETF performance across distinct market regimes and compare theoretical with empirical compounding effects. We begin with simulated, synthetic LETF returns and then compare them to realized returns from actual LETF products, adjusting for real-world frictions such as fees and tracking errors.

\subsection{Market Regimes and Compounding Effects}
Below, we focus on six key periods that capture diverse market regimes: the financial crisis (October~2007 to March 2009), its subsequent post-crisis recovery (April 2009 to March 2013), a period of sideways market (February 2014 to September~2015), the COVID-19 pandemic (February~2020 to March 2020), the post-COVID recovery (April 2020 to December 2021), and the~2022 bear market. 
Each timeframe represents a different market phase, allowing us to assess how LETFs behave under sharply declining, recovering, or stagnant markets. 

\subsubsection{Theoretical Compounding Effects via Synthetic LETFs.}
Table~\ref{table: Combined Theoretical Compounding Effects of SPY and QQQ} presents theoretical compounding effects for \emph{synthetic} LETFs.  Here, theoretical volatility refers to the effect of synthetic LETFs constructed using the benchmark ETFs \texttt{SPY} and \texttt{QQQ} with different leverage ratios $\beta \in \{-3, -2, -1, 2, 3\}$, assuming perfect replication without tracking error or fees.  Later in this section, we will compare this with the \textit{empirical} compounding effect, computed by using the actual LETFs data.

The financial crisis (October 2007 to March 2009) and the subsequent post-crisis recovery (April~2009 to March 2013) provide contrasting phases of sharp decline and sustained growth. During the crisis, the S\&P~500 index lost more than 50\% of its value, falling from over 1500 to below 700 points. In response, the Federal Reserve introduced quantitative easing, fostering a low-interest-rate environment that gradually fueled economic recovery. By 2013, the index had fully recovered to pre-crisis levels. As shown in Table~\ref{table: Combined Theoretical Compounding Effects of SPY and QQQ}, positive-leverage synthetic LETFs exhibit positive compounding effects in both recovery periods, particularly strong for $\beta = 3$. For inverse LETFs ($\beta < 0$), the compounding effects during the financial crisis are predominantly negative for \texttt{QQQ}, while \texttt{SPY} shows mixed results with $\beta = -1$ exhibiting a small positive compounding effect, but negative effects for $\beta = -2$ and $\beta = -3$. This suggests that LETFs tend to outperform expectations in strongly trending markets, making them attractive for trend following~strategies.

However, in mean-reverting markets, the effects differ by market regime. During the Sideways Market period, positive-$\beta$ synthetic LETFs consistently show negative compounding effects for both ETFs. During the COVID-19 pandemic period,~\texttt{QQQ} exhibits negative compounding effects across all leverage ratios, while \texttt{SPY} shows predominantly negative effects with a negligible positive exception for $\beta = 3$. The compounding effect is generally strongest in magnitude for highly leveraged LETFs ($\beta = \pm3$) and becomes more pronounced in high-volatility environments. The robustness check using \texttt{QQQ} confirms these findings, with stronger deviations due to sector volatility.

These results largely support Hypothesis {\bf H1}: In upward-trending regimes (2009--13, 2020--21), positive-$\beta$ synthetic LETFs exhibit consistently positive compounding effects. For downward-trending regimes (2007--09), with positive-$\beta$ LETFs showing positive compounding effects as expected, but inverse LETFs showing mixed results that partially support the hypothesis. The data also strongly confirms Hypothesis {\bf H2}: In the Sideways Market period, positive-$\beta$ LETFs for both ETFs exhibit negative compounding effects. Finally, our findings provide clear evidence for Hypothesis {\bf H3}: Across all market regimes, compounding effects are consistently largest in magnitude for the highest leverage ratios ($\beta = \pm3$), with extreme deviations occurring during high-volatility periods such as the Post-COVID Recovery and Financial Crisis.

\begin{table}[htbp]
	\scriptsize
	\centering
	\caption{Theoretical Compounding Effects of Portfolios Across Six Market Regimes}
	\label{table: Combined Theoretical Compounding Effects of SPY and QQQ}
	\begin{tabular}{l c rrrrr}
		\toprule
		& & \multicolumn{5}{c}{Leverage Ratio ($\beta$)} \\
		\cmidrule(lr){3-7}
		Market Condition & Benchmark ETF & -3x & -2x & -1x & 2x & 3x \\
		\midrule
		\multirow{2}{*}{Financial Crisis} 
		& \texttt{SPY} & -0.733 & -0.144 & 0.034 & 0.160 & 0.475 \\
		& \texttt{QQQ} & -0.883 & -0.293 & -0.035 & 0.105 & 0.346 \\
		\addlinespace
		\multirow{2}{*}{Post-Crisis Recovery} 
		& \texttt{SPY} & \textbf{2.344} & \textbf{1.349} & 0.515 & 0.651 & \textbf{1.863}\\
		& \texttt{QQQ} & \textbf{3.011} & \textbf{1.770} &0.696 & \textbf{1.005}& \textbf{3.038} \\
		\addlinespace
		\multirow{2}{*}{Sideways Market} 
		& \texttt{SPY} & -0.016 & -0.016 & -0.008 & -0.018 & -0.064 \\
		& \texttt{QQQ} & 0.117 & 0.048 & 0.011 & -0.007 & -0.043 \\
		\addlinespace
		\multirow{2}{*}{COVID-19 Pandemic} 
		& \texttt{SPY} & -0.332 & -0.141 & -0.037 & -0.007 & 0.005 \\
		& \texttt{QQQ} & -0.398 & -0.189 & -0.057 & -0.034 & -0.076 \\
		\addlinespace
		\multirow{2}{*}{Post-COVID Recovery} 
		& \texttt{SPY} & \textbf{2.034} & \textbf{1.177} & 0.459 & 0.756 & \textbf{2.670} \\
		& \texttt{QQQ} & \textbf{2.657} & \textbf{1.561} & 0.618 & \textbf{1.047} & \textbf{3.654} \\
		\addlinespace
		\multirow{2}{*}{2022 Bear Market} 
		& \texttt{SPY} & -0.251 & -0.105 & -0.027 & 0.003 & 0.011 \\
		& \texttt{QQQ} & -0.187 & -0.018 & 0.018 & 0.066 & 0.214 \\
		\bottomrule
	\end{tabular}
	\begin{tablenotes}
		
		\item Note: Values in \textbf{bold} represent substantial outperformance (defined as greater than $1.0$).
	\end{tablenotes}
\end{table}

\subsubsection{Empirical Compounding Effects and Robustness Analysis.}
We now conduct our empirical analysis using historical LETF data as presented in Table~\ref{table: empirical-volatility-effects for SPY and QQQ}. For the \texttt{SPY} ETF benchmark, we examine four leveraged ETFs:  \texttt{SDS} $(\beta = -2)$, \texttt{SH} $(\beta = -1)$, \text{SSO} $(\beta = +2)$, and \texttt{SPXL} $(\beta = +3)$. As a robustness check, we also analyze the \texttt{QQQ} ETF with its corresponding leveraged ETFs: \texttt{SQQQ} $(\beta = -3)$, \texttt{QID} $(\beta = -2)$, \texttt{QLD} $(\beta =2)$, \texttt{TQQQ} $(\beta = 3)$.

Table~\ref{table: empirical-volatility-effects for SPY and QQQ} shows a consistent sign match as predicted by the theoretical synthetic LETFs case in Table~\ref{table: Combined Theoretical Compounding Effects of SPY and QQQ} with a certain deviation. The divergence between theoretical and empirical compounding effects increases with $|\beta|$, especially during volatile periods like the COVID-19 crash. Note that transaction costs and fees account for approximately 0.8 to 1.0 percentage points per annum of the observed underperformance, while bid-ask spreads, slippage, and tracking error might contribute additional deviations; see \cite{abdi2017simple} for a simple estimation of bid-ask spread.

In particular, the COVID-19 pandemic onset (February 2020 to March 2020) caused a rapid~34\% drop in the S\&P 500 as global lockdowns disrupted economic activity. Swift monetary policy interventions enabled a quick recovery, with the index reaching new highs by late 2021. Interestingly, during the crash, inverse ETF portfolios underperformed expectations. This indicates that LETFs may fail to deliver the expected outperformance in highly volatile, mean-reverting markets due to erratic price movements.
These results validate Hypothesis {\bf H3} and align with our theoretical predictions: Higher leverage ratios ($\beta = 3$) exhibit greater deviations from theoretical predictions, particularly in high-volatility environments like the COVID-19 pandemic~period.

The 2022 bear market was marked by a 20\% decline in S\&P 500, driven by rising inflation and geopolitical uncertainties, see \cite{bouri2023global} and news released by \cite{CNBCnews}. 
As shown in Table~\ref{table: empirical-volatility-effects for SPY and QQQ}, empirical LETF performance during this period resembled that of the COVID-19 pandemic period.
The compounding effect remained weak due to increased market uncertainty, and as expected, inverse LETFs continued to underperform due to tracking error accumulation.

The sideways market (2014–2015) presents a unique challenge for LETFs, as frequent rebalancing results in negative compounding effects. While most portfolios exhibited the negative compounding effects, an exception arises for \texttt{SQQQ}, which yielded a small positive $\mathtt{CE}$. This reflects net upward drift in \texttt{QQQ} over the window. Frequent rebalancing leads to systematic buying high and selling low, which amplifies losses over time. 
Although mean-reverting markets generally induce negative compounding effects due to rebalancing drag, directional drift---even if modest---can yield a net positive effect, as observed in \texttt{SQQQ} during 2014–15. Thus, \textbf{H2} holds in \textit{expectation} but may be violated in certain drifting but volatile environments.

A robustness check using \texttt{QQQ}-based LETFs, see Table~\ref{table: empirical-volatility-effects for SPY and QQQ}, reveals similar trends but with stronger deviations. The tech-heavy Nasdaq-100 exhibited higher volatility, leading to greater tracking errors and magnified deviations in leveraged products.

\subsubsection{Summary.}
Tables~\ref{table: Combined Theoretical Compounding Effects of SPY and QQQ} and \ref{table: empirical-volatility-effects for SPY and QQQ} are directionally consistent: they agree on when $\texttt{CE}$ should be positive versus negative across $\beta$ and market regimes. The empirical $\texttt{CE}$  is systematically smaller in magnitude (due to fees, financing costs, and imperfect tracking), and some leverage ratios do not exist for the full sample.  
The robustness check using \texttt{QQQ} confirms these findings while highlighting stronger deviations in tech-heavy indices.  
In summary, the empirical results support our theoretical hypotheses:
{\bf H1} is confirmed: LETFs outperform in strongly trending regimes, regardless of direction, due to path-dependent compounding.
{\bf H2} holds broadly but admits exceptions: While most LETFs exhibit negative compounding effects in the sideways market as predicted, inverse LETFs occasionally outperform in flat markets with sufficient directional drift, as observed with \texttt{SQQQ} during 2014-15. This nuance enhances our understanding of when the volatility decay hypothesis may be counterbalanced by other factors.
{\bf H3} is robustly validated across all regimes: Empirical $\mathtt{CE}$ magnitudes scale with leverage and volatility, with the largest deviations occurring in high-volatility periods.

\begin{table}[htbp]
	\centering
	\scriptsize
	\caption{Empirical Compounding Effects of Portfolios Across Six Market Regimes}
	\label{table: empirical-volatility-effects for SPY and QQQ}
	\begin{tabular}{ l c rrrrr }
		\toprule
		& & \multicolumn{5}{c}{Leverage Ratio ($\beta$) and Product Ticker} \\
		\cmidrule(lr){3-7}
		Market Condition & Benchmark ETF & $-3\times$ & $-2\times$ & $-1\times$ & $2\times$ & $3\times$ \\
		& & (\texttt{SQQQ}) & (\texttt{SDS/QID}) & (\texttt{SH}) & (\texttt{SSO/QLD}) & (\texttt{SPXL/TQQQ}) \\
		\midrule
		\multirow{2}{*}{Financial Crisis} 
		& \texttt{SPY} & --- & 0.053 & -0.050 & 0.145 & --- \\
		& \texttt{QQQ} & -1.219 & -0.307 & --- & 0.085 & --- \\
		\addlinespace 
		\multirow{2}{*}{Post-Crisis Recovery} 
		& \texttt{SPY} & --- & \textbf{1.342} & 0.500 & 0.467 & \textbf{1.389} \\
		& \texttt{QQQ} & \textbf{3.189} & \textbf{1.762} & --- & 0.787 & --- \\
		\addlinespace 
		\multirow{2}{*}{Sideways Market} 
		& \texttt{SPY} & --- & -0.021 & -0.018 & -0.043 & -0.114 \\
		& \texttt{QQQ} & 0.112 & 0.041 & --- & -0.039 & -0.096 \\
		\addlinespace 
		\multirow{2}{*}{COVID-19 Pandemic} 
		& \texttt{SPY} & --- & -0.156 & -0.039 & -0.015 & -0.019 \\
		& \texttt{QQQ} & -0.424 & -0.198 & --- & -0.040 & -0.087 \\
		\addlinespace 
		\multirow{2}{*}{Post-COVID Recovery} 
		& \texttt{SPY} & --- & \textbf{1.173} & 0.451 & 0.666 & \textbf{2.434} \\
		& \texttt{QQQ} & \textbf{2.656} & \textbf{1.557} & --- & 0.936 & \textbf{3.370} \\
		\addlinespace 
		\multirow{2}{*}{2022 Bear Market} 
		& \texttt{SPY} & --- & -0.051 & -0.002 & -0.024 & -0.014 \\
		& \texttt{QQQ} & -0.109 & 0.035 & --- & 0.051 & 0.198 \\
		\bottomrule
	\end{tabular}

	\begin{tablenotes}
		
		\item Note: Missing data points (---) indicate values not available for the specific ETF at the given leverage factor.  For example,~\texttt{TQQQ} wasn't launched until February 2010, so data is unavailable for the Financial Crisis period. Similarly, \texttt{SPXL} launched in November 2008; no continuous daily series exists for the full crisis window. Positive values indicate outperformance relative to the target multiple, while negative values indicate underperformance. Values in \textbf{bold} represent substantial outperformance~($> 1.0$).
	\end{tablenotes}
\end{table}

\subsection{Time-Varying Compounding Effect: A Rolling Window Analysis}

While the preceding sections demonstrate that return autocorrelation is a key determinant of LETF performance, the evolution of this relationship over time remains to be characterized. To address this, we examine how the realized compounding effect (CE) varies across market regimes and assess its connection to short-run return persistence. Specifically, we compute 60-, 90-, and~120-day rolling-window estimates of compounding effects for a range of LETFs based on \texttt{SPY} and~\texttt{QQQ}, spanning multiple leverage ratios $\beta \in \{-2, -1, 2, 3\}$.

Figures~\ref{fig: 60-Day Rolling Compounding Effect from 2009 to 2024 (SPY)}--\ref{fig: 120-Day Rolling Compounding Effect from 2010 to 2024 (QQQ)} display the rolling compounding effects from 2009 (\texttt{SPY}) and 2010 (\texttt{QQQ}) to~2024 where the shaded regions represent major market events, including oil crash (mid 2015 to early~2016), and COVID-19 crash (early 2020 to late 2021), and inflation-driven bear market (early-2022).

Several systematic patterns emerge. First, positive CE is observed during persistent upward-trending markets---notably the post-2009 recovery, 2016--2017, and the 2020--2021 bull run---particularly for $\beta = 2$ and $\beta = 3$. In contrast, negative CE is concentrated in volatile or directionless periods, such as 2014--2015 and 2022, where rebalancing erodes returns. These patterns are consistent with the theoretical results in Section~\ref{section: Compounding Effect Analysis}.

\begin{figure}[htbp]
	\centering
	\includegraphics[width=0.9\textwidth]{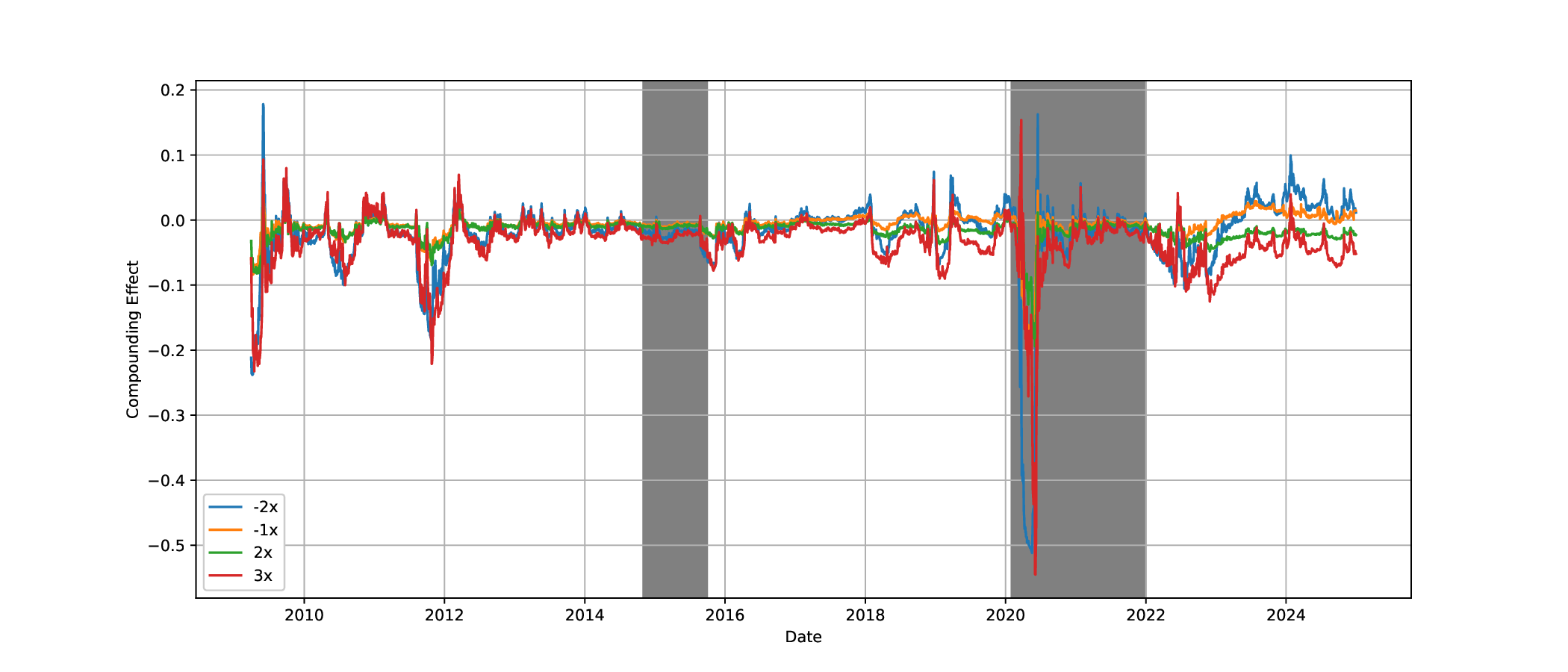}
	\caption{60-Day Rolling Compounding Effect from 2009 to 2024 (\texttt{SPY})}
	\label{fig: 60-Day Rolling Compounding Effect from 2009 to 2024 (SPY)}
\end{figure}

\begin{figure}[htbp]
	\centering
	\includegraphics[width=0.9\textwidth]{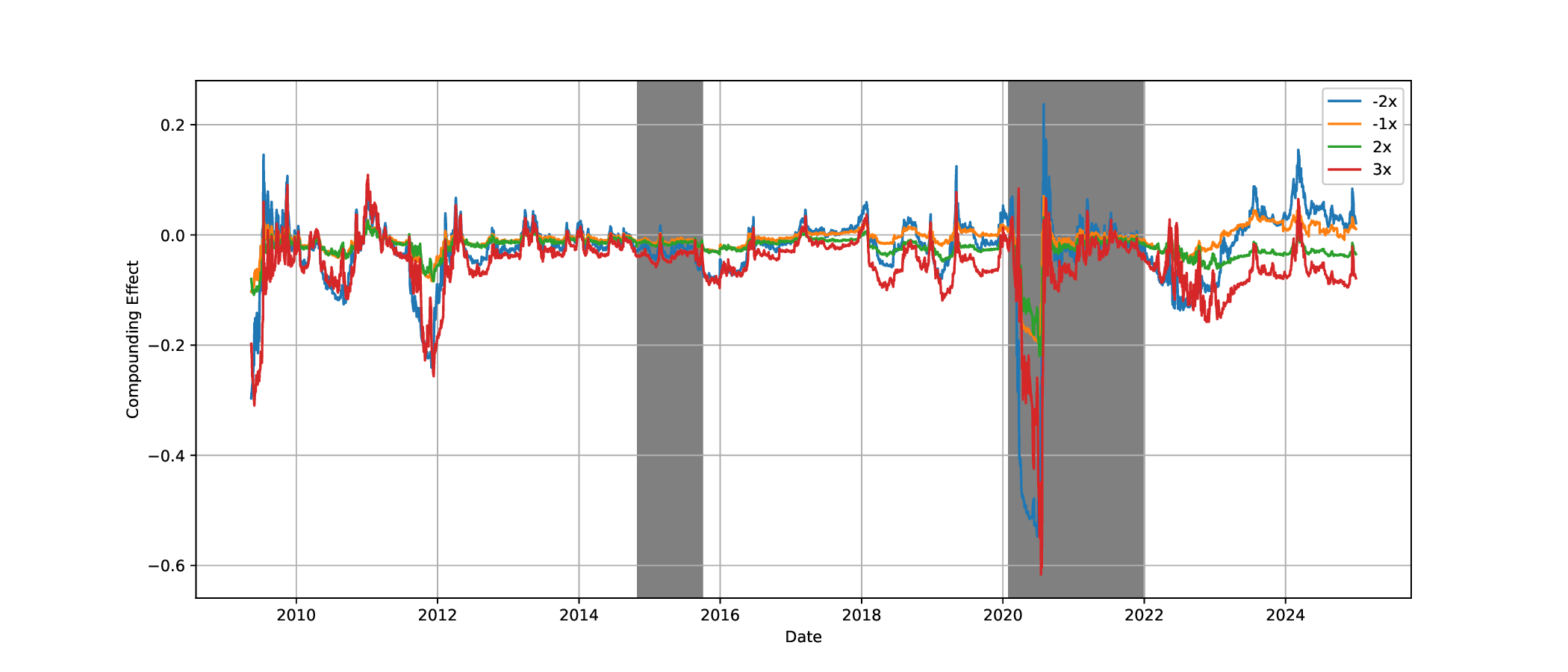}
	\caption{90-Day Rolling Compounding Effect from 2009 to 2024 (\texttt{SPY})}
	\label{fig: 90-Day Rolling Compounding Effect from 2009 to 2024 (SPY)}
\end{figure}

\begin{figure}[htbp]
	\centering
	\includegraphics[width=0.9\textwidth]{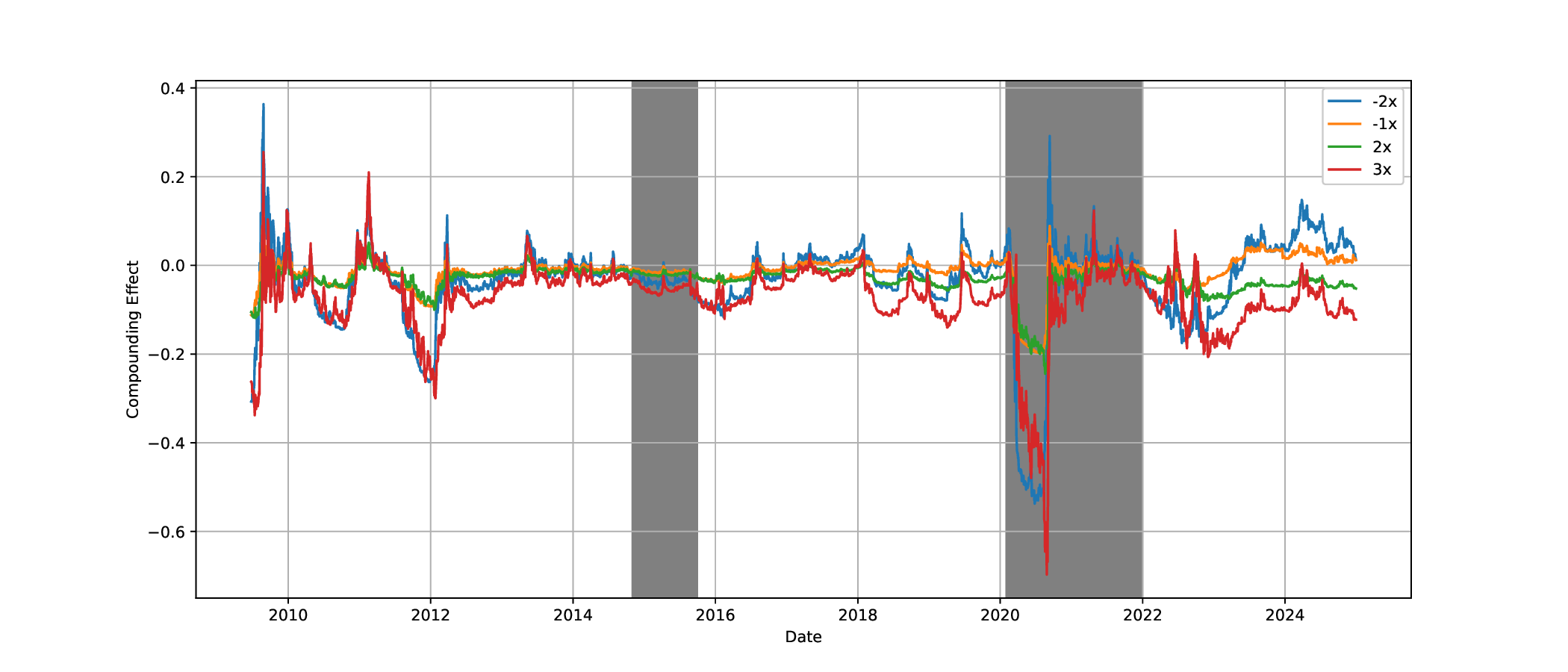}
	\caption{120-Day Rolling Compounding Effect from 2009 to 2024 (\texttt{SPY})}
	\label{fig: 120-Day Rolling Compounding Effect from 2009 to 2024 (SPY)}
\end{figure}

\begin{figure}[htbp]
	\centering
	\includegraphics[width=0.9\textwidth]{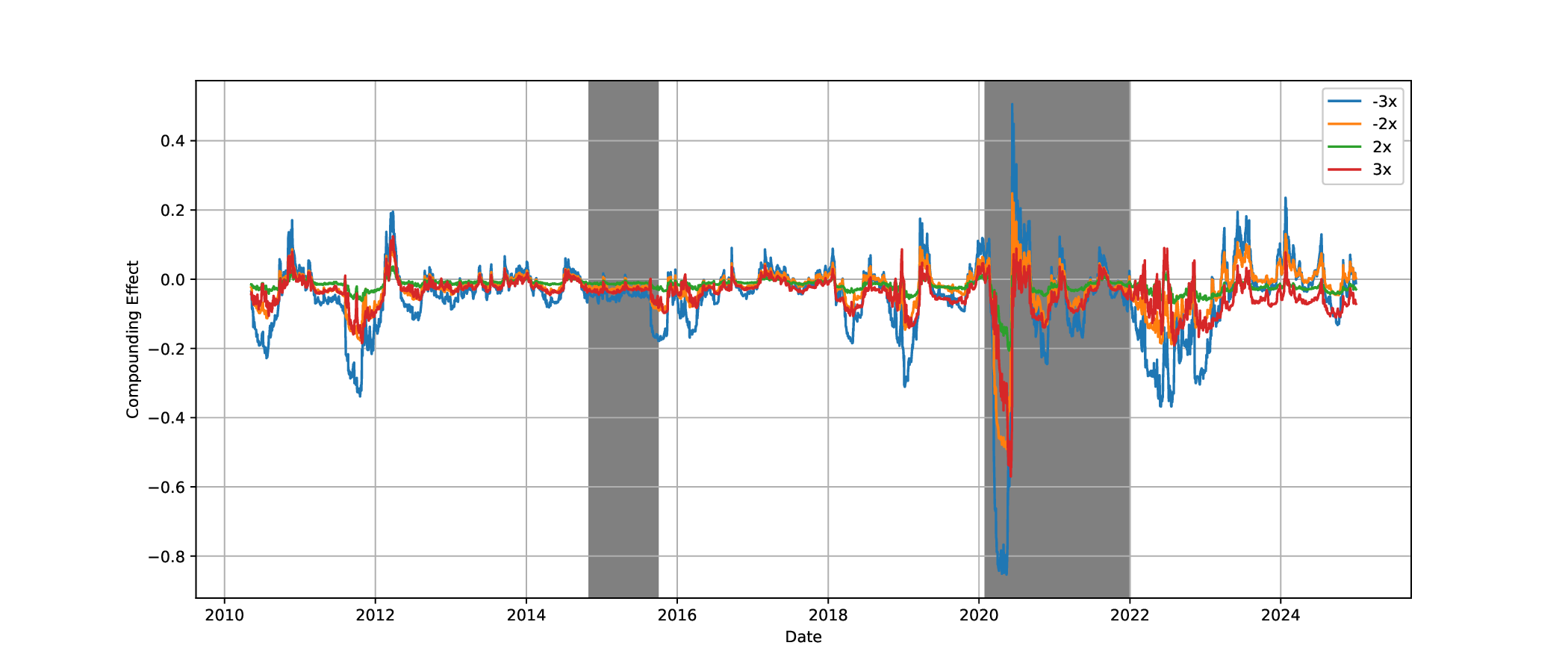}
	\caption{60-Day Rolling Compounding Effect from 2010 to 2024 (\texttt{QQQ})}
	\label{fig: 60-Day Rolling Compounding Effect from 2010 to 2024 (QQQ)}
\end{figure}

\begin{figure}[htbp]
	\centering
	\includegraphics[width=0.9\textwidth]{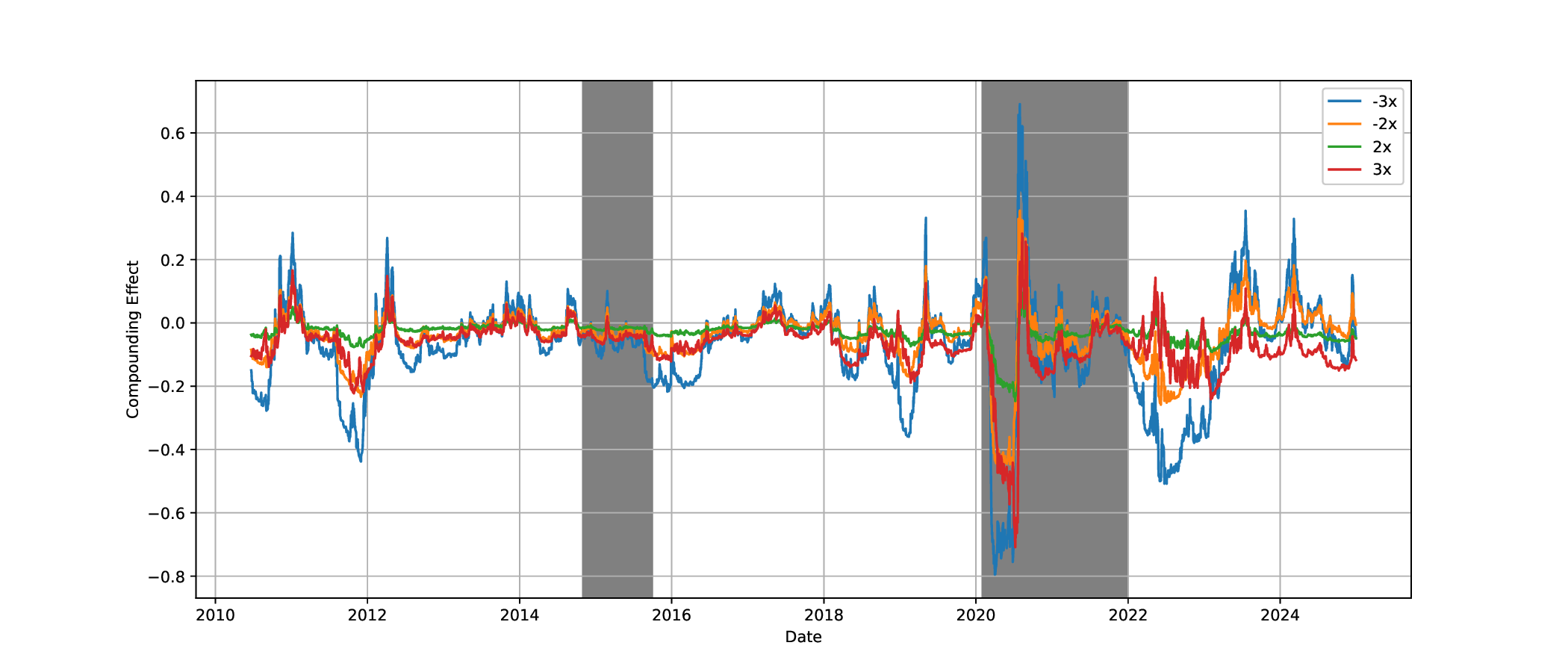}
	\caption{90-Day Rolling Compounding Effect from 2010 to 2024 (\texttt{QQQ})}
	\label{fig: 90-Day Rolling Compounding Effect from 2010 to 2024 (QQQ)}
\end{figure}

\begin{figure}[htbp]
	\centering
	\includegraphics[width=0.9\textwidth]{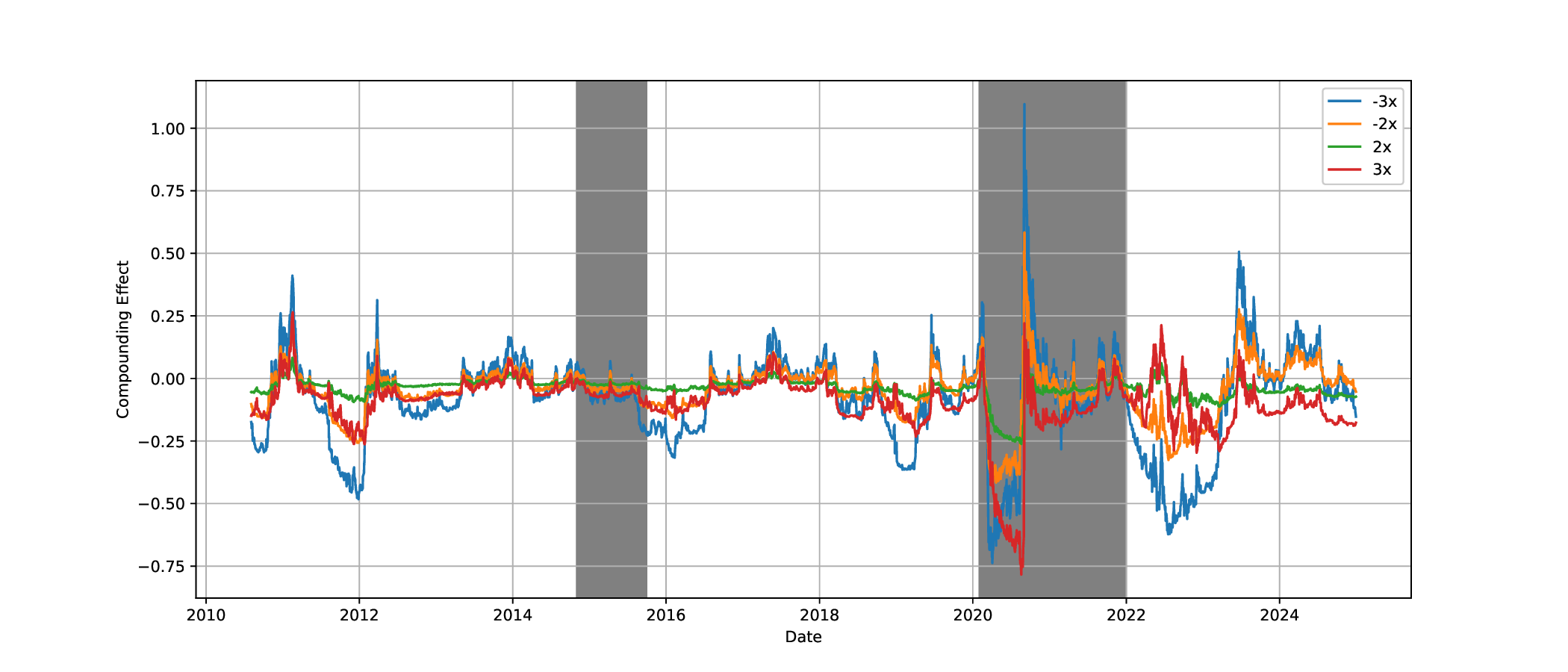}
	\caption{120-Day Rolling Compounding Effect from 2010 to 2024 (\texttt{QQQ})}
	\label{fig: 120-Day Rolling Compounding Effect from 2010 to 2024 (QQQ)}
\end{figure}

To examine the role of short-term momentum, we compute rolling AR(1) coefficients over the same intervals, presented in Figures~\ref{fig: 60-Day Rolling AR(1) Coefficients from 2009 to 2014 (SPY)}--\ref{fig: 120-Day Rolling AR(1) Coefficients from 2010 to 2014 (QQQ)}. These plots reveal fluctuations in $\phi$, often ranging between $-0.4$ and $+0.6$. Periods with elevated AR(1) coefficients, such as 2010–2011, 2016–2017, and early 2020, correspond closely to intervals of strongly positive compounding effect. Conversely, during stretches of low or negative autocorrelation---such as late 2015 or mid-2022---compounding effects decay~sharply. 

Comparing \texttt{SPY} and \texttt{QQQ}, we observe that \texttt{QQQ}-based LETFs generally exhibit more volatile and extreme $\mathtt{CE}$ fluctuations, as well as higher-amplitude AR(1) dynamics. This is consistent with the higher baseline volatility of the Nasdaq-100 and more frequent regime shifts in its constituent firms. For instance, \texttt{QQQ} $\mathtt{CE}$ spikes during both the 2020 pandemic rebound and late 2021 coincide with sharp increases in $\phi$, reinforcing the amplifying role of trend-following dynamics under high~leverage.

Overall, the rolling-window analysis reveals that LETF compounding behavior is not static, but dynamically linked to the prevailing autocorrelation environment. Positive short-run autocorrelation enhances compounding gains by aligning leverage with trending returns, while mean-reverting or noisy regimes diminish performance. These empirical observations further validate the central theoretical result of this paper: return persistence is a critical driver of LETF performance beyond volatility alone.

\begin{figure}[htbp]
	\centering
	\includegraphics[width=0.9\textwidth]{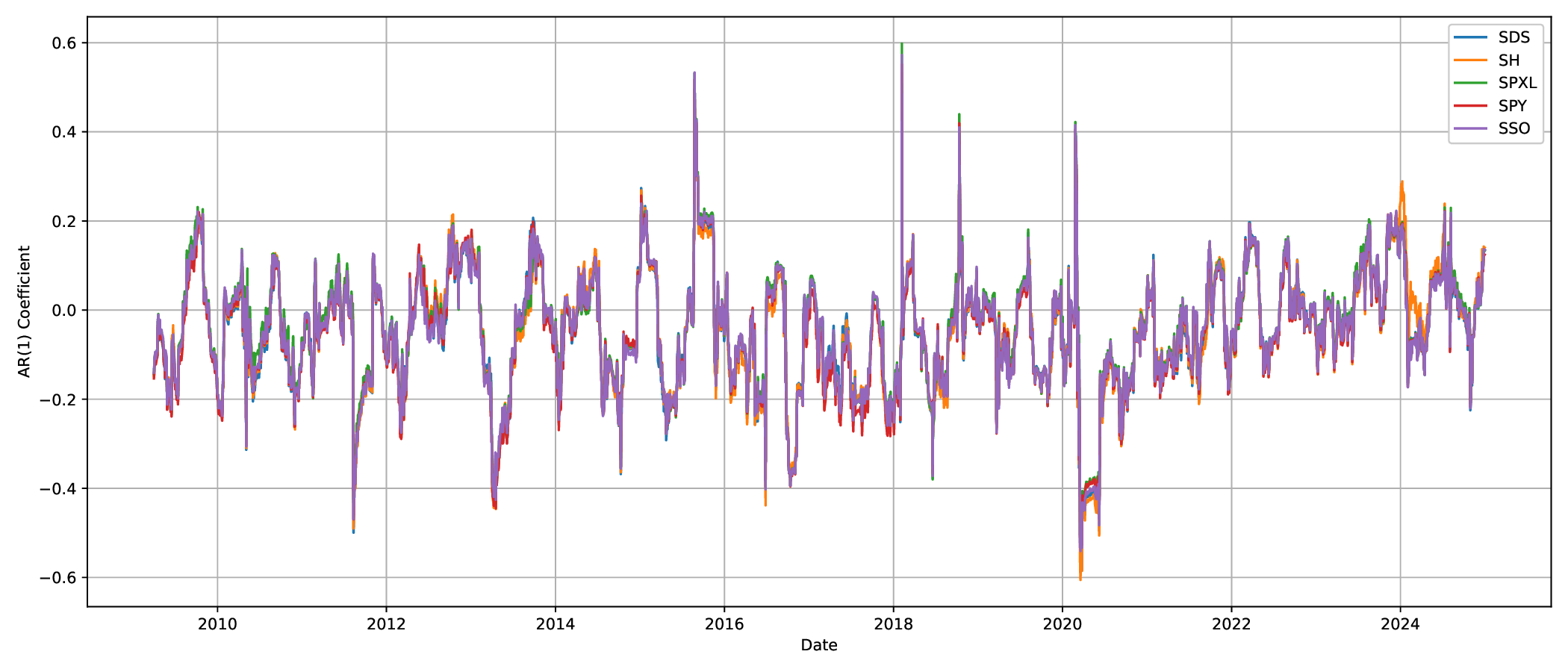}
	\caption{60-Day Rolling AR(1) Coefficients from 2009 to 2024 (\texttt{SPY})}
	\label{fig: 60-Day Rolling AR(1) Coefficients from 2009 to 2014 (SPY)}
\end{figure}

\begin{figure}[htbp]
	\centering
	\includegraphics[width=0.9\textwidth]{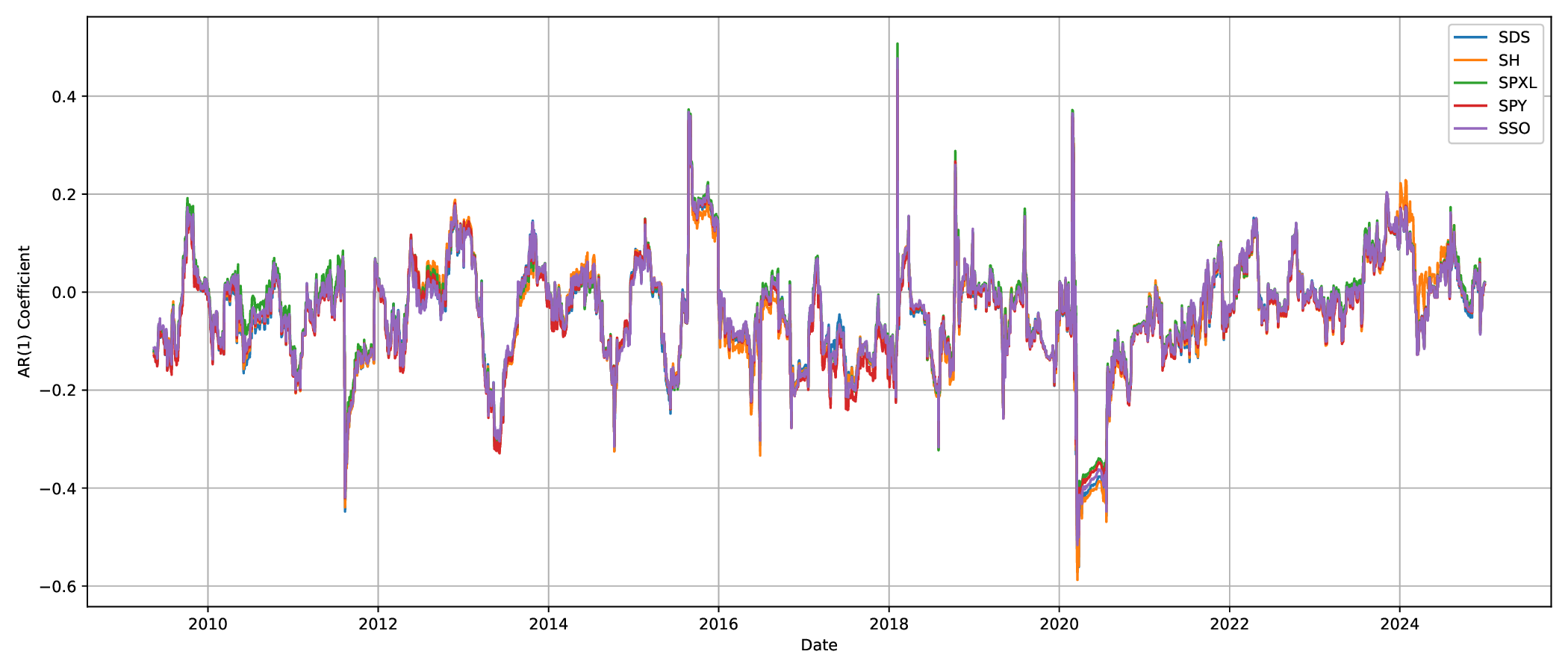}
	\caption{90-Day Rolling AR(1) Coefficients from 2009 to 2024 (\texttt{SPY})}
	\label{fig: 90-Day Rolling AR(1) Coefficients from 2009 to 2014 (SPY)}
\end{figure}

\begin{figure}[htbp]
	\centering
	\includegraphics[width=0.9\textwidth]{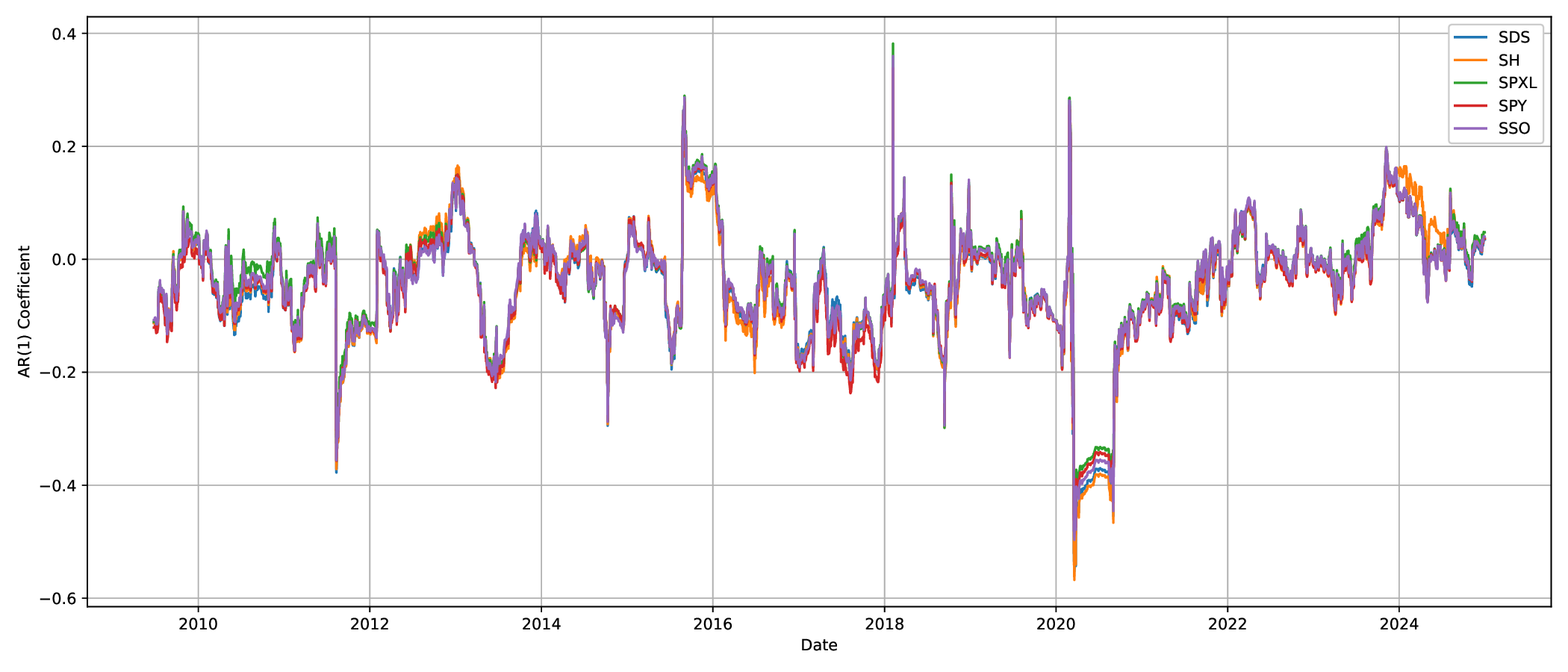}
	\caption{120-Day Rolling AR(1) Coefficients from 2009 to 2014 (\texttt{SPY})}
	\label{fig: 120-Day Rolling AR(1) Coefficients from 2009 to 2014 (SPY)}
\end{figure}

\begin{figure}[htbp]
	\centering
	\includegraphics[width=0.9\textwidth]{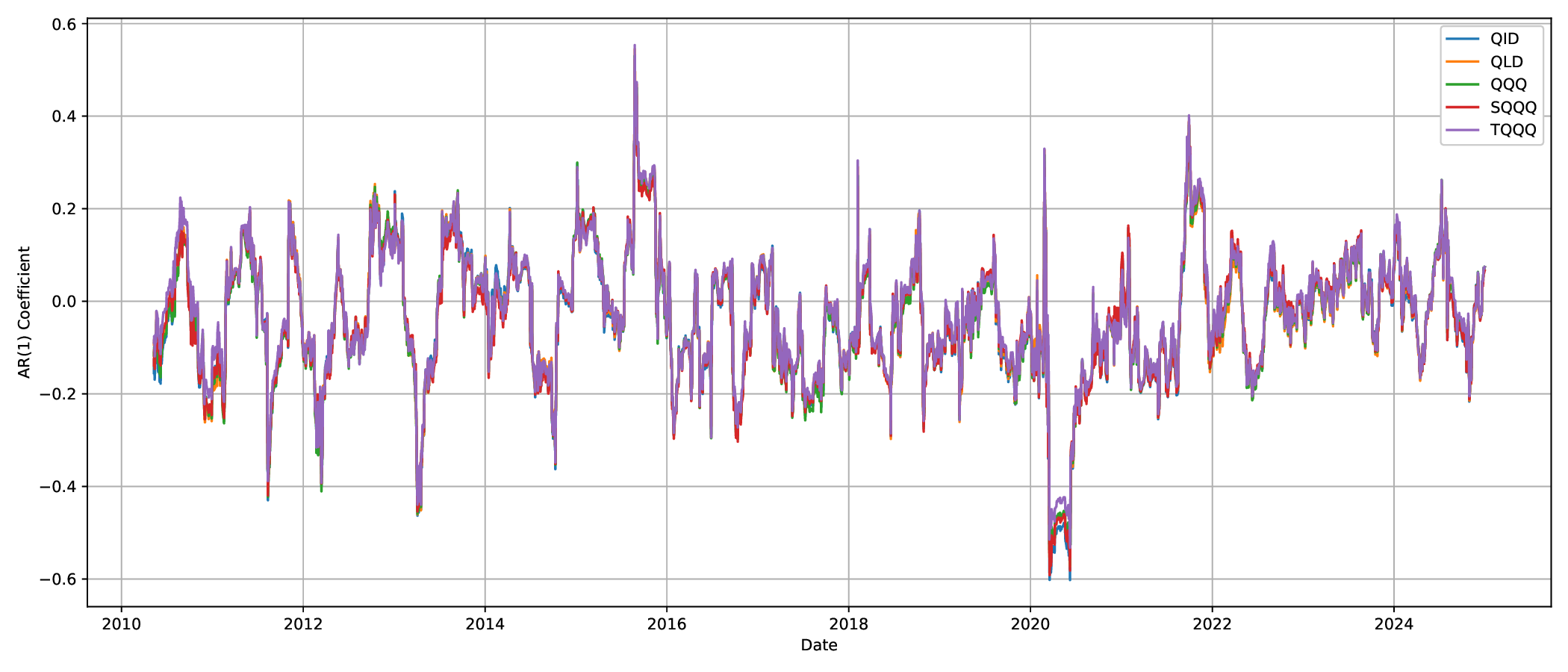}
	\caption{60-Day Rolling AR(1) Coefficients from 2010 to 2024 (\texttt{QQQ})}
	\label{fig: 60-Day Rolling AR(1) Coefficients from 2010 to 2014 (QQQ)}
\end{figure}

\begin{figure}[htbp]
	\centering
	\includegraphics[width=0.9\textwidth]{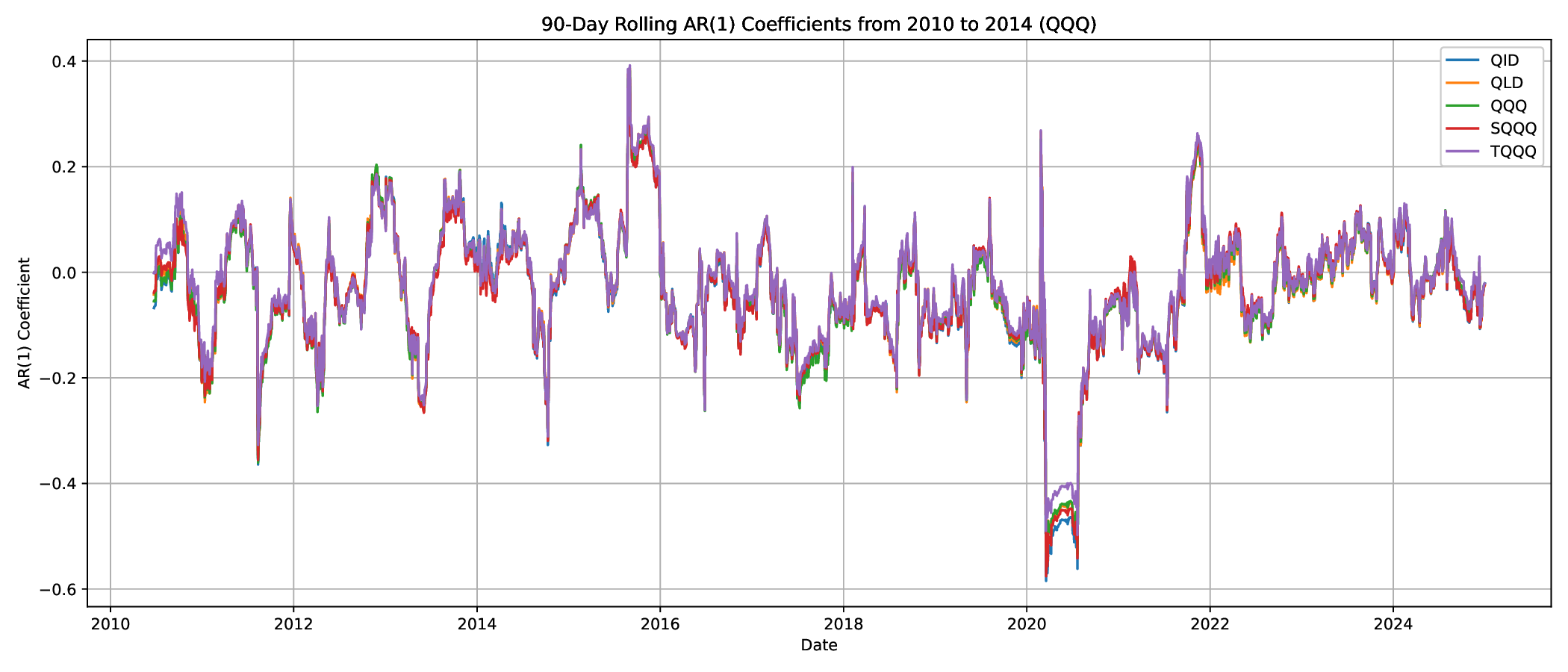}
	\caption{90-Day Rolling AR(1) Coefficients from 2010 to 2024 (\texttt{QQQ})}
	\label{fig: 90-Day Rolling AR(1) Coefficients from 2010 to 2014 (QQQ)}
\end{figure}

\begin{figure}[htbp]
	\centering
	\includegraphics[width=0.9\textwidth]{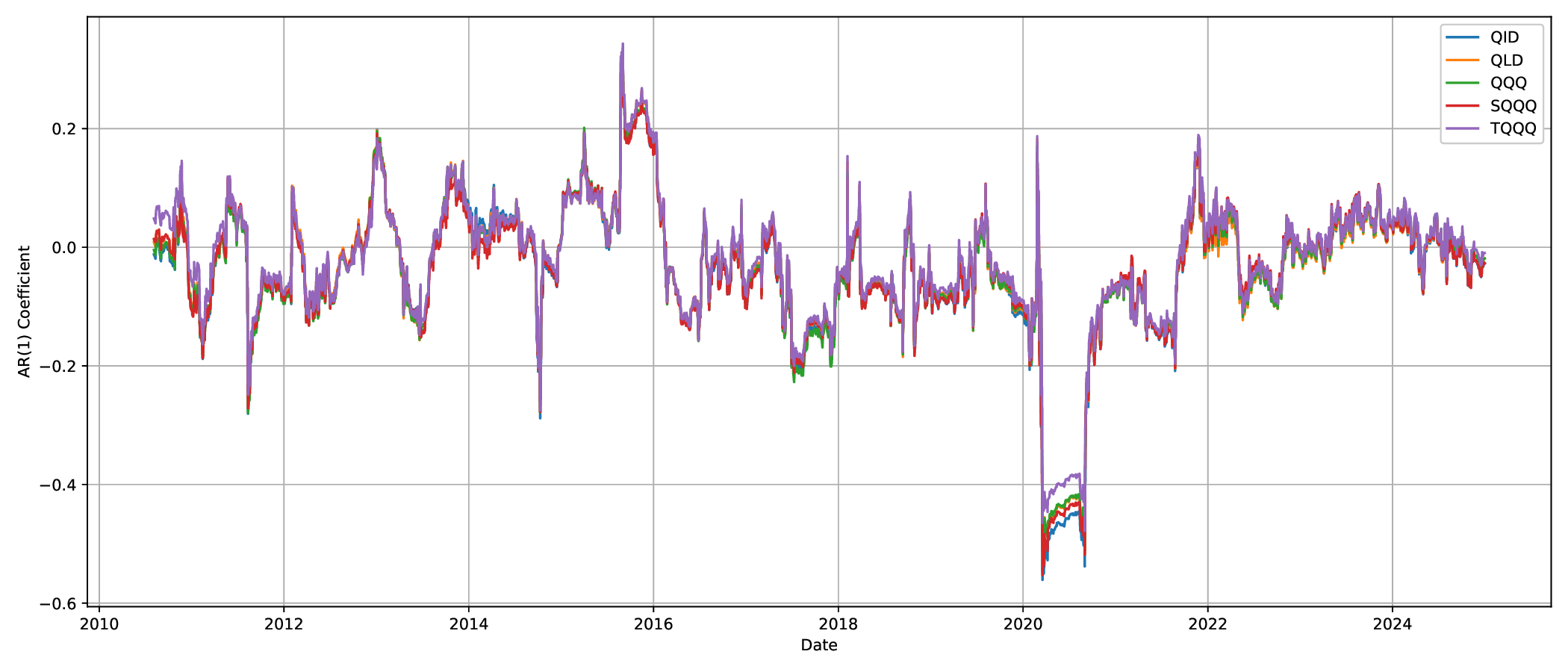}
	\caption{120-Day Rolling AR(1) Coefficients from 2010 to 2024 (\texttt{QQQ})}
	\label{fig: 120-Day Rolling AR(1) Coefficients from 2010 to 2014 (QQQ)}
\end{figure}

\section{Concluding Remarks}\label{section: conclusion}

This paper explores the compounding effect of LETFs under various return dynamics, volatility levels, and rebalancing frequencies.  Our central finding is that return autocorrelation and return dynamics---not the volatility alone---determine whether LETFs outperform or underperform their targets. In particular, momentum improves compounding, while mean reversion undermines it, with these effects magnified under frequent rebalancing.

We develop a unified theoretical framework encompassing AR(1), AR-GARCH, and continuous-time regime-switching models. These models reveal that volatility clustering interacts with autocorrelation to shape LETF behavior, particularly under persistent shocks or regime transitions. This insight refines the conventional volatility-drag narrative by emphasizing the joint role of autocorrelation and volatility persistence.

In the benchmark environments with independent returns, LETFs tend to outperform their target multiples on average, and changes in volatility or rebalancing frequency have minimal impact. In contrast, in serially correlated markets, LETF performance is sensitive to the autocorrelation, return dynamics, and the rebalancing interval. Daily rebalancing enhances returns in trending markets, whereas weekly or monthly rebalancing mitigates losses in oscillating or mean-reverting~regimes.

Empirical validation using about 20 years of \texttt{SPY} and \texttt{QQQ} data across diverse market regimes confirms these predictions. LETFs exhibit positive compounding during directional trends, but underperform in sideways markets where rebalancing erodes returns. For trend-following investors, daily-rebalanced LETFs offer superior performance. In contrast, when autocorrelation is weak or negative, longer rebalancing intervals reduce performance decay. Future work could extend this framework by incorporating rough volatility models or non-Markovian regime-switching dynamics to capture finer features of real-world return processes.

\bibliographystyle{apalike}
\bibliography{references}

\medskip
\appendix

	\section{Technical Proofs} \label{appendix: technical proofs}
	This appendix collects some technical proofs of the paper.
	
	\subsection{Proofs in Section~\ref{section: Preliminaries}} \label{appendix: technical proofs in Preliminaries}

	\begin{proof}[Proof of Theorem~\ref{theorem: Compounding Effect Derivation}]
		We begin by noting that the ETF’s daily return is  
		$
		X_t^{\rm ETF} =  X_t
		$
		and LETF’s daily return is to include fees and tracking error:
		$
		X_t^{\rm LETF} =  \beta X_t - f + e_t,
		$
		where \(\beta \in \mathbb{Z} \setminus \{0\}\) is the leverage ratio, \(f\) is a daily fee (incorporating expense and transaction costs), and \(e_t\) is a tracking error with \(|e_t| \le M\) and \(\mathbb{E}[e_t] = \mathbb{E}[e_t X_s] = \mathbb{E}[e_t e_s] = 0\) for \(t \ne s\). Assume also \(|X_t|, |f| \le M\), for some \(M \ll 1\).
		Define:
		\[
		A_t := \beta X_t - f + e_t, \quad B_t := X_t.
		\]
		Hence, the compounding effect is:
		$
		\mathtt{CE}_n := R_n^{{\rm LETF}} - \beta R_n^{{\rm ETF}}
		$ 
		where the cumulative returns over $n$ days are:
		\[
		R_n^{\mathrm{LETF}} = \prod_{t=1}^n (1 + A_t) - 1 \; \text{ and } \; R_n^{\mathrm{ETF}} = \prod_{t=1}^n (1 + B_t) - 1.
		\]
		Using the fact that 
		$
		\prod_{t=1}^n (1 + Z_t) = 1 + \sum_{t=1}^n Z_t + \sum_{1 \le t < s \le n} Z_t Z_s + \mathfrak{R}_n$ with $ \mathbb{E}[| \mathfrak{R}_n| ] \le C^\prime n^3 M^3
		$ for some $C^\prime >0$,
		it follows that $R_n^{\mathrm{LETF}}$ and $R_n^{\mathrm{ETF}}$ satisfies
		\begin{align*}
			R_n^{\rm LETF} &= \sum_{t=1}^n A_t + \sum_{1 \le t < s \le n} A_t A_s + \mathfrak{R}_n^{(A)}, \\
			R_n^{\rm ETF}  &= \sum_{t=1}^n B_t + \sum_{1 \le t < s \le n} B_t B_s + \mathfrak{R}_n^{(B)},
		\end{align*}
		and therefore
		\begin{align} \label{eq: CE_n formula}
			\mathtt{CE}_n = R_n^{\rm LETF} - \beta R_n^{\rm ETF} =   \sum_{t=1}^n (A_t - \beta B_t) + \sum_{1 \le t < s \le n} (A_t A_s - \beta B_t B_s) + \mathfrak{R}_n,
		\end{align}
		where $ \mathfrak{R}_n := \mathfrak{R}_n^{(A)} - \beta \mathfrak{R}_n^{(B)} $, and \(\mathbb{E}[|\mathfrak{R}_n|] \le C n^3 M^3 \) for some constant $C >0$ by triangle inequality.
		Note that the first-order term in Equation~\eqref{eq: CE_n formula} can be computed as follows:
		\[
		\sum_{t=1}^n (A_t - \beta B_t) = \sum_{t=1}^n (\beta X_t - f + e_t - \beta X_t) = -nf + \sum_{t=1}^n e_t.
		\]
		Taking expectations and using \(\mathbb{E}[e_t] = 0\), we obtain:
		$
		\mathbb{E}\left[ \sum_{t=1}^n (A_t - \beta B_t) \right] = -nf.
		$
		Next, for $t\neq s$, we expand:
		\begin{align*}
			A_t A_s 
			&= (\beta X_t - f + e_t)(\beta X_s - f + e_s)\\
			&= \beta^2 X_t X_s - \beta f (X_t + X_s) + f^2 + \beta X_t e_s + \beta X_s e_t - f(e_t + e_s) + e_t e_s.
		\end{align*}
		Thus,
		$
		A_t A_s - \beta B_t B_s 
		= \beta(\beta - 1) X_t X_s - \beta f (X_t + X_s) + f^2 + \beta (X_t e_s + X_s e_t) - f (e_t + e_s) + e_t e_s.
		$
		Taking the expectation yields
		\[
		\mathbb{E}[A_t A_s - \beta B_t B_s] = \beta(\beta - 1) \gamma_{|t-s|} - 2\beta f \mu + f^2.
		\]
		where the equality holds by using the facts that \(\mathbb{E}[X_t X_s] = \gamma_{|t-s|}\) from strict stationarity,  \(\mathbb{E}[e_t] = \mathbb{E}[e_t X_s] = \mathbb{E}[e_t e_s] = 0\) for \(t \ne s\), and \(\mu := \mathbb{E}[X_t]\).
		Summing over all \(1 \le t < s \le n\), we group terms by lag \(k = s - t\), which contributes \((n-k)\) terms per lag. This gives:
		\[
		\mathbb{E} \left[	\sum_{1 \le t < s \le n}A_t A_s - \beta B_t B_s \right]
		= \sum_{1 \le t < s \le n} \mathbb{E}[A_t A_s - \beta B_t B_s]
		= \sum_{k=1}^{n-1} (n-k) \left[ \beta(\beta - 1) \gamma_k - 2\beta f \mu + f^2 \right].
		\]
		Putting everything together:
		\[
		\mathbb{E}[\mathtt{CE}_n]
		= -nf + \sum_{k=1}^{n-1} (n-k) \left[ \beta(\beta - 1) \gamma_k - 2\beta f \mu \right]+ \binom{n}{2} f^2 + O(n^3 M^3). \qedhere
		\]
	\end{proof}

	\begin{proof}[Proof of Lemma~\ref{lemma: Positive Expected Compounding Effect}] 
		To prove the desired positivity, we proceed with a proof by induction. 
		Firstly, for $n=1$, note that
		$
		\mathbb{E}[\mathtt{CE}_1] = 	\mathbb{E}\left[ R_{1}^{\rm LETF, \beta} - \beta R_{1}^{\rm ETF} \right] = [(1 + \beta \mu)^{1} - 1] - \beta[(1+ \mu)^{1} - 1 ] =0.
		$
		Now, assume the statement holds for~$n = k \geq 1$. That is,
		$
		[(1 + \beta \mu)^{k}-1]-\beta[(1+\mu)^{k}-1] \geq 0, 
		$
		which is equivalent to
		$
		(1 + \beta\mu)^{k} + \beta - 1 \geq \beta(1+\mu)^{k}.
		$
		We must show that the statement holds for $n = k+1$. Indeed, observe that
		\begin{align}
			(1 + \beta\mu)^{k+1}-1 
			=&(1+\beta\mu)(1+\beta\mu)^{k}-1\nonumber\\
			=&(1+\beta\mu)[(1+\beta\mu)^{k}+\beta-1]-\beta-\beta^{2}\mu+\beta\mu\nonumber\\
			\geq&(1+\beta\mu)[\beta(1+\mu)^{k}]-\beta-\beta^{2}\mu+\beta\mu\nonumber
		\end{align}
		where inequality is held by invoking the inductive hypotheses. Hence, it follows that
		\begin{align}
			(1 + \beta\mu)^{k+1}-1 
			\geq&(1+\beta\mu)[\beta(1+\mu)^{k}]-\beta-\beta^{2}\mu+\beta\mu\nonumber\\
			=&(1+ \mu +(\beta-1)\mu)[\beta(1+\mu)^{k}]-\beta-\beta^{2}\mu+\beta\mu\nonumber\\
			=&\beta(1+\mu)^{k+1}+\beta(\beta-1)\mu(1+\mu)^{k}-\beta-\beta^{2}\mu+\beta\mu\nonumber\\
			=&\beta[(1+\mu)^{k+1}-1]+\beta(\beta-1)\mu(1+\mu)^{k}+\beta\mu(1-\beta)\nonumber\\
			=&\beta[(1+\mu)^{k+1}-1]+\beta\mu(\beta-1)[(1+\mu)^{k}-1].
		\end{align}
		Note that if $\beta \mu (\beta-1)[(1+\mu)^{k}-1] \geq 0$, then we are done. 
		Therefore, to complete the proof, it remains to prove $\beta \mu( \beta-1)[(1+\mu)^{k}-1] \geq 0$. Set an auxiliary function 
		$
		f(\beta, \mu) := \beta\mu(\beta-1)[(1+\mu)^{k}-1],
		$
		then, we consider the following four cases:
		
		\emph{Case 1.} When $\beta > 1$ and $\mu \geq 0$, it follows that $\beta \mu \geq 0$, $\beta - 1 > 0$, and $[(1 + \mu)^{k+1} - 1] \geq 0$, leading to $f(\beta, \mu) \geq 0$.
		
		\emph{Case 2.} When $\beta > 1$ and $\mu < 0$, we have $\beta \mu < 0$, $\beta - 1 > 0$, and $[(1 + \mu)^{k+1} - 1] < 0$, resulting in $f(\beta, \mu) > 0$.
		
		\emph{Case 3.} When $\beta \leq -1$ and $\mu \geq 0$, it follows that $\beta \mu \leq 0$, $\beta - 1 < 0$, and $[(1 + \mu)^{k+1} - 1] \geq 0$, which implies $f(\beta, \mu) \geq 0$.
		
		\emph{Case 4.} When $\beta \leq -1$ and $\mu < 0$, we obtain $\beta \mu > 0$, $\beta - 1 < 0$, and $[(1 + \mu)^{k+1} - 1] < 0$, leading to $f(\beta, \mu) \geq 0$ again.
		
		Combined with the four cases above, we conclude that $f(\beta, \mu) \geq 0$ for all cases. Thus, the given statement holds for any $n \in \mathbb{N}$ and $\beta \in \mathbb{Z} \setminus \{0,1\}$, which proves that the expected compounding effect is always positive.
	\end{proof}

	\medskip
	\subsection{Compounding Effect for Serially Correlated AR(1) Returns} \label{appendix: compounding effect for serially correlated returns}
	Below, we shall work with the AR(1) model for $\{X_t: t \geq 1\}$.
	Let $\beta \in \mathbb{Z} \setminus \{0,1\}.$ For the two-period case, we have that:  $(i)$ If $\phi >0$, then  $\mathbb{E}[ \mathtt{CE}_2 ] >0$ and
	$(ii)$ If $\phi <0$, then $ \mathbb{E}[ \mathtt{CE}_2] < 0$.
	To see this, we begin by noting that
	\begin{align*}
		\mathtt{CE}_2 
		&= [(1+\beta X_1)(1+\beta X_2) - 1] - \beta [ (1+X_1)(1+X_2)-1] \\
		&=\beta( \beta - 1) X_1 X_2 .
	\end{align*}
	Taking the expectation on both sides yields:
	$
	\mathbb{E}[ \mathtt{CE}_2  ]
	= \beta( \beta - 1) \mathbb{E}[ X_1 X_2]
	= \beta (\beta - 1) \phi \frac{\sigma}{1-\phi^2}.  
	$
	Note that $\sigma > 0$, and for $\beta \in \mathbb{Z} \setminus \{0, 1\}$, the product $\beta (\beta - 1) >0$.  Therefore, to determine the sign of the expected difference, it suffices to check the sign of the $\phi$. Specifically, for $\phi >0$, then  $\mathbb{E}[  \mathtt{CE}_2  ] > 0$.
	On the other hand, for $\phi < 0$, then $\mathbb{E}[  \mathtt{CE}_2 ] < 0$. The following proof extends this result to the multi-period case.

	\begin{proof}[Proof of Lemma~\ref{lemma: n-period case for serially correlated AR(1) returns}]   
		
		Given that $\beta \in \mathbb{Z}\setminus\{0, 1\}$, we observe that
		\begin{align*}
			\mathbb{E} [\mathtt{CE}_n ] 
			&= \mathbb{E} \left[\left(\prod_{i=1}^n (1+\beta X_i) - 1\right) - \beta \left( \prod_{i=1}^n (1+X_i) - 1\right)\right] \\
			&= \mathbb{E}  \left[(\beta^2 - \beta)\sum_{i \neq j} X_i X_j + (\beta^3 - \beta)\sum_{i \neq j \neq k} X_i X_j X_k + \cdots + (\beta^n - \beta)\sum_{i \neq j \neq \cdots \neq n} X_i X_j X_k \cdots X_n \right]\\
			&= \mathbb{E}  \left[(\beta^2 - \beta)\sum_{i \neq j} X_i X_j + \text{higher-order terms in $X_i$} \right].
		\end{align*}
		Since $ |X_i| < m$ for some $m \in (0,1)$, the first quantity dominates, and the higher-order product term is of smaller order compared to the second-order term for sufficiently small returns. Therefore, we~obtain
		\begin{align}
			\mathbb{E} [ \mathtt{CE}_n] 
			&\approx \mathbb{E} \left[ (\beta^2 - \beta)\sum_{i \neq j} X_i X_j\right] \notag \\
			&= (\beta^2 - \beta)\left( \sum_{i=1}^{n-1} \mathbb{E}[X_i X_{i+1}] + \sum_{i=1}^{n-2} \mathbb{E}[X_i X_{i+2}] + \cdots + \sum_{i=1}^{1} \mathbb{E} [X_i X_{i+(n-1)}]\right), \notag \\
			&= (\beta^2 - \beta)\frac{\sigma^2}{1-\phi^2}\left( \sum_{i=1}^{n-1} \phi + \sum_{i=1}^{n-2} \phi^2 + \cdots + \sum_{i=1}^{1} \phi^{n-1}\right), \notag \\
			&= (\beta^2 - \beta)\frac{\sigma^2}{1-\phi^2}\phi\left[(n-1) + (n-2)\phi + \cdots + \phi^{n-2}\right] \gtrless 0 \text{ if } \phi \gtrless 0  \label{ineq: proof of lemma 3}
		\end{align}
		where $|\phi| < 1$. Note that since $\beta \in \mathbb{Z} \setminus \{0,1\}$, it follows that $(\beta^2 - \beta) = \beta (\beta - 1) >0$.
		To see inequality~\eqref{ineq: proof of lemma 3}, let $Q:=(n-1)+(n-2)\phi+\cdots+\phi^{n-2}$.
		Then, for sufficiently large $n$, $\mathbb{E} [ \mathtt{CE}_n ] = (\beta^2 - \beta)\frac{\sigma^2}{1-\phi^2}\phi Q$,  it is readily verified that
		$
		\phi Q - Q = \phi + \phi^2 + \cdots + \phi^{n-2} + \phi^{n-1} - (n-1).
		$
		Therefore, for all $|\phi| < 1$,
		\begin{align*}
			Q &=  \frac{n -1+\phi+\phi^2+\cdots+\phi^{n-2}+\phi^{n-1}}{1-\phi} \\
			&= \frac{(1-1)+(1-\phi)+(1-\phi^2)+\cdots+(1-\phi^{n-2})+(1-\phi^{n-1})}{1-\phi} > 0.
		\end{align*}
		Therefore, $\mathbb{E} [ \mathtt{CE}_n ] \gtrless 0$ if $\phi \gtrless 0$.
	\end{proof}

	\subsection{Compounding Effect for Continuous-Time Model} \label{appendix: Compounding Effect for Continuous-Time Return Model}

	\begin{proof}[Proof of Equation~\eqref{eq: LETF L_t_regime_switching}] 
		Begin by observing that
		\begin{align*}
			\frac{d L_t }{L_t} 
			&= \beta \frac{dS_t}{S_t}   - fdt
			= (\beta \mu_{Z_t}  - f)dt + \beta \sigma_{Z_t} dW_t.
		\end{align*}
		Given that $ f(L_t) := \log L_t $, It\^o's Lemma states:
		\begin{align*}
			d \log(L_t )
			& = \left( \frac{1}{ L_t} L_t( \beta \mu_{Z_t} - f) - \frac{1}{2} \frac{1}{ L_t^2} (L_t^2 \beta^2 \sigma_{Z_t}^2) \right) dt + \frac{1}{ L_t} L_t \beta \sigma_{Z_t} dW_t\\
			& = \left( (\beta \mu_{Z_t}-f) - \frac{1}{2}  \beta^2 \sigma_{Z_t}^2 \right) dt +  \beta \sigma_{Z_t}dW_t.
		\end{align*}
		Hence, it has a solution
		\begin{align*}
			L_t 
			&= L_0 \exp\left( \int_0^t \left(  \beta \mu_{Z_s} - \frac{1}{2}  \beta^2 \sigma_{Z_s}^2 \right) ds - ft +  \beta \int_0^t \sigma_{Z_s} dW_s \right)
		\end{align*}
		and the proof is complete.
	\end{proof}

	\begin{proof}[Proof of Theorem~\ref{theorem: expected compounding effect for diffusion prices_regime}]   
		Note that 
		$$
		\mathbb{E}[\mathtt{CE}_t] = \mathbb{E}[\exp(Y_t)] - \beta \mathbb{E}[\exp(X_t)] + (\beta - 1),
		$$
		where $Y_t  = \int_0^t \left(  \beta \mu_{Z_s} - \frac{1}{2}  \beta^2 \sigma_{Z_s}^2 \right) ds - ft +  \beta \int_0^t \sigma_{Z_s} dW_s$ and $X_t = \int_0^t \left( \mu_{Z_s}  - \frac{1}{2}   \sigma_{Z_s}^2 \right) ds + \int_0^t \sigma_{Z_s} dW_s.$
		We approximate the expectation using a regime-mixture approximation based on the occupation measure:
		Let $\pi_j(t) := \mathbb{E}\left[ \frac{1}{t} \int_0^t \mathbb{1}_{\{Z_s = j\}} ds \right]$ for each $j=1,\dots,M$. This is the expected proportion of time the process spends in regime $j$ over $[0,t]$.
		If $Z_t = j$, then
		\begin{align*}
			X_t^{(j)} 
			&=\int_0^t \left( \mu_j  - \frac{1}{2}   \sigma_j^2 \right) ds + \int_0^t \sigma_j dW_s\\
			& =  \left( \mu_j  - \frac{1}{2}   \sigma_j^2 \right)t + \sigma_j W_t \sim \mathcal{N}\left(\left( \mu_j  - \frac{1}{2}   \sigma_j^2 \right)t, \sigma_j^2 t \right)
		\end{align*}
		and
		\begin{align*}
			Y_t^{(j)} 
			&= \int_0^t \left(  \beta \mu_{j} - \frac{1}{2}  \beta^2 \sigma_{j}^2 \right) ds - ft +  \beta \int_0^t \sigma_{j} dW_s \\
			&= (\beta \mu_j -\frac{1}{2}  \beta^2 \sigma_{j}^2 - f )t + \beta \sigma_j W_t \sim \mathcal{N}\left( (\beta \mu_j -\frac{1}{2}  \beta^2 \sigma_{j}^2 - f )t, \,  \beta^2 \sigma_j^2 t  \right).
		\end{align*}
		Hence, using the moment generating function technique, it follows that
		\begin{align*}
			\mathbb{E}[ \exp({X_t}^{(j)}) ] = \exp\left(  \mu_j   t\right) 
			\text{ and }
			\mathbb{E}[ \exp({Y_t}^{(j)} )] = \exp\left( (\beta \mu_j  - f )t \right)
		\end{align*}
		Because $Z_t$ switches between regimes over time, the full $Y_t$ is a weighted combination of regime-specific dynamics. The exact expectation $\mathbb{E}[\exp(X_t)]$ and $\mathbb{E}[\exp(Y_t)]$ can be  approximated\footnote{In principle, one can compute $\mathbb{E}[\exp(X_t)]$ and $\mathbb{E}[\exp(Y_t)]$ condition on given $Z_s = j$. Then invoke the Feynman-Kac theorem to solve the resulting PDEs, e.g., see \cite{cvitanic2004introduction}. However, solving the PDEs can be complex. } as follows:
		\begin{align*}
			\mathbb{E}[\exp(X_t)] 
			&\approx \sum_{j=1}^{M} \pi_j(t) \cdot \mathbb{E}[\exp ( X_t^{ (j) })] = \sum_{j=1}^{M} \pi_j(t) \cdot \exp\left( \mu_j  t \right)\\
			\mathbb{E}[\exp(Y_t)] 
			& \approx \sum_{j=1}^{M} \pi_j(t) \cdot \mathbb{E}[\exp ( Y_t^{ (j) })] = \sum_{j=1}^{M} \pi_j(t) \cdot \exp\left( (\beta \mu_j  - f )t \right).
		\end{align*}
		Therefore, the expected compounding effect is obtained:
		\[
		\mathbb{E}[\mathtt{CE}_t] \approx \sum_{j=1}^{M} \pi_j(t) \left[  \exp\left( (\beta \mu_j  - f )t \right) - \beta \exp\left(  \mu_j t \right) \right] + (\beta - 1). \qedhere
		\]
	\end{proof}

	\begin{lemma}[Expected Occupation Measure] \label{lemma: expected occupation measure}
		The expected occupation measure
		$\sum_{i=1}^M \pi_j(t) = 1$ for all $ t \geq 0$.
	\end{lemma}
	
	\begin{proof}[Proof of Lemma~\ref{lemma: expected occupation measure}]
		For every fixed sample point $\omega \in \Omega$ and any time $s$, the Markov chain is exactly one of the $M$ states. Hence, for all $s \geq 0$ and $\omega \in \Omega$,
		$
		\sum_{j=1}^M \mathbb{1}_{\{Z_s(\omega) = j\}} = 1 .
		$
		Hence, with the aid of the Fubini-Tonelli theorem, we have
		\begin{align*}
			\sum_{j=1}^M \pi_j(t) 
			&= \sum_{j=1}^M  \mathbb{E}\left[ \frac{1}{t} \int_0^t \mathbb{1}_{\{Z_s = j\}} ds \right]\\ 
			&= \mathbb{E}\left[  \sum_{j=1}^M \frac{1}{t} \int_0^t \mathbb{1}_{\{Z_s = j\}} ds \right]\\
			&= \mathbb{E}\left[   \frac{1}{t} \int_0^t \sum_{j=1}^M \mathbb{1}_{\{Z_s = j\}} ds \right]\\
			&= \mathbb{E}\left[   \frac{1}{t} \int_0^t 1 ds \right]
			= \mathbb{E}\left[   1 \right] = 1. \qedhere
		\end{align*}
	\end{proof}

	\begin{proof}[Proof of Theorem~\ref{theorem: Sign of Expected CE under Regime Switching}]
		This is an immediate consequence of the approximation:
		\[
		\mathbb{E}[\mathtt{CE}_t] \approx \sum_j \pi_j(t) \Phi_j(t; \beta, f) + (\beta - 1)
		\]
		where $\Phi_j(t; \beta, f)$ depends on both $\beta$ and $\mu_j$. The sign of the convex combination $\sum_j \pi_j(t) \Phi_j(t; \beta, f)$ determines the sign of CE, relative to the offset $(1 - \beta)$.
	\end{proof}

	\begin{proof}[Proof of Corollary~\ref{corollary: Nonnegative CE in Single Regime}]
		This follows directly from Theorem~\ref{theorem: Sign of Expected CE under Regime Switching}. 
		Fix $f=0$.	Note that for $t=0$, $	\mathbb{E}[\mathtt{CE}_t] =0$. Next, for $t >0$, 
		Take $x := \mu_j t \in \mathbb{R}$ where $x \geq 0$ if $\mu_j >0$ and $x \leq 0$ if $\mu_j <0$, and $x=0$ if $\mu = 0$. Hence, it suffices to show that
		$
		h_\beta(x):= \exp\left( \beta x  \right)  - \beta \exp(x) + (\beta - 1) \geq 0
		$ for all $x \in \mathbb{R}$.
		We now consider two cases:
		
		\textit{Case 1}. For $\beta > 1$, set the weights
		$
		\lambda_1 := \frac{1}{\beta}, \lambda_2 := 1 - \frac{1}{\beta}
		$
		then $\lambda_1 + \lambda_2 =1$ and $\lambda_1, \lambda_2 > 0.$ Since the exponential function $\exp(\cdot)$ is convex, applying the Jensen's inequality yields
		$
		\lambda_1 e^{\beta x } + \lambda_2 e^0 \geq e^{\lambda_1 \beta x + \lambda_2 \cdot 0} = e^x.
		$
		Multiplying~$\beta$ we obtain
		$
		e^{\beta x}-\beta e^{x}+(\beta-1) \geq 0
		$ for all $x \in \mathbb{R}$, 
		which shows that $\mathbb{E}[\mathtt{CE}_t] \geq 0$ for all $ t \geq 0$ when~$\beta >1.$
		
		\textit{Case 2}. For $\beta <0$, Take $\beta := -k$ with $k \in \mathbb{N}$. Then, we obtain
		$
		h_k(x) := e^{-k x} + k e^{x}-(k+1).
		$
		Note that the first derivative: $h_k'(x) = -ke^{-kx} + k e^x$, which vanishes at $x=0.$ Moreover, its second-order derivative: $
		h_k''(x) = k^2 e^{-kx} + k e^x >0.
		$
		Hence, $h_k$ is strictly convex and $x=0$ is the unique global minimizer of $h_k$. Value at the minimum is
		$
		h_k(0) = 1+ k - (k+1) = 0.
		$
		Since the minimum value is $0$, we have $h_k(x)\geq 0$ for every~$x \in \mathbb{R}$. Hence, it also implies that $\mathbb{E}[\mathtt{CE}_t] \geq 0$ for all $t$.
		
		Lastly, we consider the degenerate case $\mu = 0$. It implies that
		$
		\mathbb{E}[\mathtt{CE}_t] = e^0 - \beta e^0 + (\beta -1) = 0,
		$
		which is desired. Therefore, in all cases, $\mathbb{E}[\mathtt{CE}_t]  \geq 0$ for all $t \geq 0$ and equality holds when $t=0$, or when~$\mu = 0.$
	\end{proof}


\end{document}